\documentclass[11pt]{amsart}
\usepackage{latexsym,amsmath,amssymb,amscd, eucal,verbatim}
\usepackage[all]{xy}
\xyoption{curve}
\xyoption{line}
\usepackage{graphicx}
\newcommand\bC{\mathbb C}
\newcommand\bQ{\mathbb Q}

\newcommand\bZ{\mathbb Z}
\newcommand\Ca{$C^*$-al\-ge\-bra}

\newcommand{\rk}{\operatorname{rk}}
\newcommand{\Tors}{\operatorname{Tors}}

\newcommand{\simple}[3]{
\xy
(0,0)*{#1}; 
(0,0)*\xycircle(3.3,3.3){-}="a"; 
(20,0)*{#2}; 
(20,0)*\xycircle(3.3,3.3){-}="b"; 
"a"; "b" **\dir{-}?(.6)*\dir{>}+(-1,5)*{#3};
\endxy
}

\newcommand{\tripleright}[4]{
\xy
(0,10)*{#1}; 
(0,10)*\xycircle(3.3,3.3){-}="a"; 
(0,-10)*{#2}; 
(0,-10)*\xycircle(3.3,3.3){-}="b"; 
(20,0)*{#3}; 
(20,0)*\xycircle(3.3,3.3){-}="c"; 
(6,0)*{}="d";
 "a";"d" **\dir{-}?(.6)*\dir{>};
 "d"; "c"**\dir{-}?(.6)*\dir{>};
 "b"; "d" **\dir{-}?(.6)*\dir{>};
(13,8)*{#4}; 
\endxy
}

\newcommand{\tripleleft}[4]{
\xy
(0,0)*{#1}; 
(0,0)*\xycircle(3.3,3.3){-}="a"; 
(20,10)*{#2}; 
(20,10)*\xycircle(3.3,3.3){-}="b"; 
(20,-10)*{#3}; 
(20,-10)*\xycircle(3.3,3.3){-}="c"; 
(14,0)*{}="d";
"a"; "d" **\dir{-}?(.6)*\dir{>};
"d"; "b" **\dir{-}?(.6)*\dir{>};
"d"; "c" **\dir{-}?(.6)*\dir{>};
(7,8)*{#4}; 
\endxy
}

\newcommand{\four}[5]{
\xy
(0,10)*{#1}; 
(0,10)*\xycircle(3.3,3.3){-}="a"; 
(0,-10)*{#2}; 
(0,-10)*\xycircle(3.3,3.3){-}="b"; 
(6,0)*{}="d";
 "a"; "d" **\dir{-}?(.6)*\dir{>};
  "b";"d" **\dir{-}?(.6)*\dir{>};
  (20,0)*{}="c";
  (26,10)*{#3}; 
(26,10)*\xycircle(3.3,3.3){-}="x"; 
(26,-10)*{#4}; 
(26,-10)*\xycircle(3.3,3.3){-}="y"; 
"c" ; "x" **\dir{-}?(.6)*\dir{>};
"c"; "y" **\dir{-}?(.6)*\dir{>};
"d"; "c" **\dir{-}; 
(13,16)*{#5};
\endxy
}

\newcommand{\bottom}[8]{
\xy
(0,0)*{#1}; 
(0,0)*\xycircle(3.3,3.3){-}="a"; 
(20,10)*{#2}; 
(20,10)*\xycircle(3.3,3.3){-}="b"; 
(20,-10)*{#3}; 
(20,-10)*\xycircle(3.3,3.3){-}="c"; 
(14,0)*{}="d";
"a"; "d" **\dir{-}?(.6)*\dir{>};
"d"; "b" **\dir{-}?(.6)*\dir{>};
"d"; "c" **\dir{-}?(.6)*\dir{>};
(7,8)*{#4}; 
(35,10)*{#5}; 
(35,10)*\xycircle(3.3,3.3){-}="x"; 
(35,-10)*{#6}; 
(35,-10)*\xycircle(3.3,3.3){-}="y"; 
(55,0)*{#7}; 
(55,0)*\xycircle(3.3,3.3){-}="z"; 
(41,0)*{}="t";
 "x";"t" **\dir{-}?(.6)*\dir{>};
 "t"; "z"**\dir{-}?(.6)*\dir{>};
 "y"; "t" **\dir{-}?(.6)*\dir{>};
(48,8)*{#8}; 
(22.7,-10)*{}="cc";
(32.5,-10)*{}="yy"; 
"cc" ; "yy" **\dir{~};
\endxy
}

\newcommand{\both}[8]{
\xy
(0,0)*{#1}; 
(0,0)*\xycircle(3.3,3.3){-}="a"; 
(20,10)*{#2}; 
(20,10)*\xycircle(3.3,3.3){-}="b"; 
(20,-10)*{#3}; 
(20,-10)*\xycircle(3.3,3.3){-}="c"; 
(14,0)*{}="d";
"a"; "d" **\dir{-}?(.6)*\dir{>};
"d"; "b" **\dir{-}?(.6)*\dir{>};
"d"; "c" **\dir{-}?(.6)*\dir{>};
(7,8)*{#4}; 
(35,10)*{#5}; 
(35,10)*\xycircle(3.3,3.3){-}="x"; 
(35,-10)*{#6}; 
(35,-10)*\xycircle(3.3,3.3){-}="y"; 
(55,0)*{#7}; 
(55,0)*\xycircle(3.3,3.3){-}="z"; 
(41,0)*{}="t";
 "x";"t" **\dir{-}?(.6)*\dir{>};
 "t"; "z"**\dir{-}?(.6)*\dir{>};
 "y"; "t" **\dir{-}?(.6)*\dir{>};
(48,8)*{#8}; 
(22.7,10)*{}="bb";
(32.5,10)*{}="dd"; 
(22.7,-10)*{}="cc";
(32.5,-10)*{}="yy"; 
"cc" ; "yy" **\dir{~};
"bb" ; "dd" **\dir{~};
\endxy
}

\DeclareMathOperator{\K}{\sf K}

\DeclareMathOperator{\KK}{\sf KK}
\DeclareMathOperator{\Hom}{Hom}

\newcommand{\HE}{{\sf HL}}

\DeclareMathOperator{\HP}{\sf HP}

\DeclareMathOperator{\End}{End}
\DeclareMathOperator{\Aut}{Aut}
\DeclareMathOperator{\Ext}{Ext}

\DeclareMathOperator{\spin}{spin}

\DeclareMathOperator{\ch}{ch}
\DeclareMathOperator{\Ind}{index}

\newcommand{\rar}{\rightarrow}
\newcommand{\bbr}{\mathbb{R}}
\newcommand{\bbR}{\mathbb{R}}
\newcommand{\bbc}{\mathbb{C}}
\newcommand{\bbC}{\mathbb{C}}

\newcommand{\cC}{\mathcal{C}}

\newcommand{\CT}{{\sf{CT}}}

\newcommand{\mbf}[1]{{\boldsymbol {#1} }}

\newcommand{\eq}{\begin{equation}}
\newcommand{\eqend}{\end{equation}}

\newbox\ncintdbox \newbox\ncinttbox
\setbox0=\hbox{$-$} \setbox2=\hbox{$\displaystyle\int$}
\setbox\ncintdbox=\hbox{\rlap{\hbox to
\wd2{\hskip-.125em\box2\relax \hfil}}\box0\kern.1em}
\setbox0=\hbox{$\vcenter{\hrule width 4pt}$}
\setbox2=\hbox{$\textstyle\int$}
\setbox\ncinttbox=\hbox{\rlap{\hbox to
\wd2{\hskip-.175em\box2\relax \hfil}}\box0\kern.1em}

\def\Dirac{{D\!\!\!\!/\,}} 
\def\dirac{{\partial\!\!\!/\,}} 

\newcommand{\twga}{C_r^*(\Gamma_g,\sigma_\theta)}

\newcommand{\complex}{{\mathbb C}} 
\newcommand{\zed}{{\mathbb Z}} 
\newcommand{\bbZ}{{\mathbb Z}}
\newcommand{\bbT}{{\mathbb T}}
\newcommand{\real}{{\mathbb R}} 
\newcommand{\RR}{{\mathbb R}} 
\newcommand{\rat}{{\mathbb Q}} 
\newcommand{\id}{{1\!\!1}} 
\def\alg{{\mathcal A}}
\def\balg{{\mathcal B}}
\def\calg{{\mathcal C}}
\def\dalg{{\mathcal D}}
\def\ealg{{\mathcal E}}
\newcommand\jalg{{\mathcal J}}

\def\hil{{\mathcal H}}
\def\bun{{\mathcal E}}

\def\cK{{\mathcal K}}

\def\K{{\sf K}}

\def\H{{\sf H}}
\def\P{{\sf P}}

\def\B{{\mathfrak B}}

\def\W{{\sf W}}
\def\HP{{\H\P}}

\def\op{{\rm o}}
\def\Id{{\rm id}}
\def\pt{{\rm pt}}
\def\Cl{{\sf Cliff}}
\def\ch{{\sf ch}}
\def\Ch{{\sf Ch}}

\def\Todd{{\sf Todd}}

\newcommand{\bbz}{{\mathbb Z}}

\def\even{{\rm even}}

\newcommand{\tr}[1]{\:{\rm tr}\,#1}

\def\be{\begin{equation}}
\def\ee{\end{equation}}
\def\bea{\begin{eqnarray}}
\def\eea{\end{eqnarray}}
\def\bd{\begin{displaymath}}
\def\ed{\end{displaymath}}

\def\dd{{\rm d}}

\def\ii{{\,{\rm i}\,}}

\makeatletter
\newdimen\normalarrayskip              
\newdimen\minarrayskip                 
\normalarrayskip\baselineskip
\minarrayskip\jot
\newif\ifold             \oldtrue            
\def\arraymode{\ifold\relax\else\displaystyle\fi} 
\def\@arrayskip{\ifold\baselineskip\z@\lineskip\z@
     \else
     \baselineskip\minarrayskip\lineskip2\minarrayskip\fi}
\def\@arrayclassz{\ifcase \@lastchclass \@acolampacol \or
\@ampacol \or \or \or \@addamp \or
   \@acolampacol \or \@firstampfalse \@acol \fi
\edef\@preamble{\@preamble
  \ifcase \@chnum
     \hfil$\relax\arraymode\@sharp$\hfil
     \or $\relax\arraymode\@sharp$\hfil
     \or \hfil$\relax\arraymode\@sharp$\fi}}
\def\@array[#1]#2{\setbox\@arstrutbox=\hbox{\vrule
     height\arraystretch \ht\strutbox
     depth\arraystretch \dp\strutbox
     width\z@}\@mkpream{#2}\edef\@preamble{\halign \noexpand\@halignto
\bgroup \tabskip\z@ \@arstrut \@preamble \tabskip\z@ \cr}%
\let\@startpbox\@@startpbox \let\@endpbox\@@endpbox
  \if #1t\vtop \else \if#1b\vbox \else \vcenter \fi\fi
  \bgroup \let\par\relax
  \let\@sharp##\let\protect\relax
  \@arrayskip\@preamble}
\makeatother

\newcommand{\beq}{\begin{eqnarray}}
\newcommand{\eeq}{\end{eqnarray}}

\newcommand{\G}{\Gamma}

\newcommand{\cL}{{\mathcal L}}
\newcommand{\cV}{{\mathcal V}}

\def\appendix#1{\addtocounter{section}{1}\setcounter{equation}{0}
\renewcommand{\thesection}{\Alph{section}}
\section*{Appendix \thesection. #1}
\protect\indent \parbox[t]{11.715cm}
}

\newcommand{\bbq}{\mathbb{Q}}

\def\ii{{\,{\rm i}\,}}

\newtheorem{theorem}{Theorem}[section]
\newtheorem{lemma}[theorem]{Lemma}
\newtheorem{cor}[theorem]{Corollary}
\newtheorem{proposition}[theorem]{Proposition}
\theoremstyle{definition}
\newtheorem{definition}[theorem]{Definition}
\theoremstyle{remark}
\newtheorem{example}[theorem]{Example}
\newtheorem{remark}[theorem]{Remark}

\numberwithin{equation}{section}

\begin{document}


\title[Noncommutative correspondences,
 duality and D-branes in bivariant $\K$-theory]
{Noncommutative correspondences,\\
 duality and D-branes in bivariant $\K$-theory}

\author{Jacek Brodzki}

\address{School of Mathematics, University of Southampton,
Southampton SO17 1BJ, UK}

\email{j.brodzki@soton.ac.uk}

\author{Varghese Mathai}

\address{Department of Pure Mathematics, University of Adelaide,
Adelaide 5005, Australia}

\email{mathai.varghese@adelaide.edu.au}

\author{Jonathan Rosenberg}

\address{Department of Mathematics, University of Maryland,
College Park, MD 20742, USA}

\email{jmr@math.umd.edu}

\author{Richard J. Szabo}

\address{Department of Mathematics and Maxwell Institute for
  Mathematical Sciences, Heriot-Watt University, Riccarton, Edinburgh
  EH14 4AS, UK}

\email{R.J.Szabo@ma.hw.ac.uk}

\begin{abstract}
We describe a categorical framework for the classification of D-branes
on noncommutative spaces using techniques from bivariant $\K$-theory
of $C^*$-algebras. We present a new description of bivariant
$\K$-theory in terms of noncommutative correspondences which is nicely
adapted to the study of T-duality in open string theory. We
systematically use the diagram calculus for bivariant $\K$-theory as
detailed in our previous paper~\cite{BMRS}. We explicitly work out our
theory for a number of examples of noncommutative manifolds.
\end{abstract}
 


\maketitle

\section*{Introduction}

Understanding the mathematical structures underpinning D-branes and
Ramond-Ramond fields in superstring theory is important for
systematically describing some of their novel properties. One of the
best known examples of this is the classification of D-brane charges
and RR-fields using techniques of $\K$-theory~\cite{MM,
  W,FW,Horava,OS99,MW}. The $\K$-theory
framework naturally accounts for certain properties of D-branes and
RR-fields that would not be realized if these objects were classified
by ordinary cohomology or homology alone. For example, it explains the
appearance of stable but non-BPS branes carrying torsion charges, and
correctly incorporates both the self-duality and quantization
conditions on RR-fields. It has also led to a variety of new
predictions concerning the spectrum of superstring theory, such as the
instability of D-branes wrapping non-contractible cycles in certain
instances due to the fact that their cohomology classes do not
``lift'' to $\K$-theory~\cite{Diaconescu:2000wy}, and the obstruction
to simultaneous measurement of electric and magnetic RR fluxes when
torsion fluxes are included~\cite{Freed:2006ya}. Moreover, certain
properties of the string theory path integral, such as worldsheet
anomalies and certain subtle phase factor contributions from the
RR-fields, are most naturally formulated within the context of
$\K$-theory~\cite{Diaconescu:2000wy,FW}.

A natural but very complicated problem is to determine the set of
possible D-branes in a given closed string background. One of the
difficulties is that many of the consistent conformally invariant
boundary conditions for open strings have no geometrical
description. Furthermore, in the presence of certain non-trivial
background supergravity form fields the worldvolume field theories of
D-branes are best described in the language of noncommutative
geometry. Previous work (\cite{BEM,MR}, among many other references)
also showed that the formulation of T-duality in the presence of
NS--NS flux requires the use of noncommutative
manifolds. The simplest examples of noncommutative spacetimes result from
applying T-duality to spacetimes compactified 
along directions which have non-trivial support of the NS--NS field.

In a previous paper~\cite{BMRS}, we took preliminary steps towards
extending the $\K$-theory classification of D-branes and RR-fields in
the commutative case to
spacetimes that are noncommutative manifolds. In particular, we
obtained an explicit formula for the charges of D-branes in
noncommutative spacetimes. The purpose of the present paper is to
elaborate on these constructions somewhat and to flesh out many
explicit examples, particularly some which arise naturally in string
theory. The presentation here is hoped to be in a form that is more
palatable to physicists. We describe the constructions of~\cite{BMRS}
by exploiting the diagram calculus developed there, which provides
both a computationally useful tool as well as a heuristic guide to some of
the more technical aspects of the formalism. The detailed proofs of
many of our previous results are deferred to~\cite{BMRS}. Instead, we
focus our attention on presenting many new examples and relating them
to standard constructions in the string theory literature.

There are various themes underlying this work, which we explain in
detail in the following sections. In many applications it is useful to
characterise the dynamics of D-branes in a categorical fashion. In
this setting, D-branes are regarded, in a certain sense, as objects in
some category. Such a point of view has been fruitful in topological
string theory and applications to homological mirror
symmetry~\cite{Aspinwall:2004jr}, wherein B-model D-branes are
regarded as objects in a derived category of coherent sheaves while
A-model D-branes are objects in a Fukawa category. The collection of
open string boundary conditions may also be regarded as a category in
the context of two-dimensional open/closed topological field
theory~\cite{Moore:2006dw}, and thereby used to explain the
structure of boundary conformal field theory and its connections to
$\K$-theory. In this paper we regard D-branes as objects in a certain
category of separable $C^*$-algebras. This is the category underlying
Kasparov's bivariant $\K$-theory (also known as $\KK$-theory). This
characterisation is closely related to the open string algebras for
different boundary conditions that arise in open string field theory,
which for boundary conditions of maximal support are Morita equivalent
via open string bimodules.

The bivariant version of $\K$-theory is useful for a variety of
reasons. Firstly, it incorporates both the $\K$-theory and
$\K$-homology descriptions of D-branes in a unified
setting. But bivariant $\K$-theory is richer than both $\K$-theory
and $\K$-homology considered individually, as it carries an
intersection product. Secondly, this product structure
provides the correct mathematical framework in which
to formulate notions of duality between generic separable
$C^*$-algebras, such as Poincar\'e duality (see
for example \cite{Moscovici}). This can be used to
explain the equivalence of the $\K$-theory and
$\K$-homology descriptions of D-brane charges.  It was also
exploited in~\cite{BMRS} to provide a categorical description of open
string T-duality using an involutive type of $\KK$-equivalence, which
includes the more familiar notion of Morita duality~\cite{SW} as a
special case in some simpler situations. This characterisation has the
virtue of extending the usual geometric notions behind T-duality to
examples involving ``non-geometric'' backgrounds. At least in some
examples, Grange and Sch\"afer-Nameki~\cite{Grange:2006es}
have proposed to view such noncommutative spacetimes as
a globally defined, open string version of Hull's T-fold 
proposal~\cite{Hull:2004in}. Thirdly, bivariant theories are
the appropriate venue
for defining topological invariants of noncommutative spaces, an
example being the Todd class defined in~\cite{BMRS}. In the following
we will calculate the Todd class for some explicit examples of
noncommutative manifolds. Finally, it is the setting that is used to
define a generic notion of $\K$-orientation, and hence to generalise
the Freed-Witten anomaly cancellation condition~\cite{FW} which
selects the consistent sets of D-branes from our category. These
formulations all culminate in a noncommutative version of the D-brane
charge vector. In this paper we will work out explicit examples (both
commutative and noncommutative) of all of these quantities.

A central result in this paper is a new description of bivariant
$\K$-theory in terms of equivalence classes of {\em noncommutative
  correspondences}. This description is nicely adapted to the study of
dualities in open string theory such as T-duality, mirror symmetry and
smooth analogs of the Fourier-Mukai transform. The commutative version
of the correspondence picture nicely unifies the geometric
construction of states of D-branes in terms of topological $\K$-cycles
and in terms of topological $\K$-theory classes in spacetime. In the generic
noncommutative setting it therefore nicely captures the point of view
of D-branes as objects in the $\KK$-theory category with the morphisms
between objects, which are the elements of the bivariant $\K$-theory
groups, providing generalizations of T-duality transformations. T-dual
D-branes are themselves realized in terms of invertible elements of
$\KK$-theory which define order two T-duality actions up to Morita
equivalence. As we discuss, via the universal coefficient theorem,
bivariant $\K$-theory provides a refinement of the $\K$-theoretic
notion of T-duality.

Let us now summarise the contents and structure of this
paper. In~\S\ref{sect:KKdiagrams} we review the definition of
bivariant $\K$-theory and bivariant cyclic theory for separable
$C^*$-algebras. We also review the diagram calculus developed in an
earlier paper~\cite{BMRS}, which is useful for manipulating the
products in these groups. \S\ref{sect:ncPD} is concerned with
understanding Poincar\'e duality in $\K$-theory and cyclic theory
using the diagram calculus. In particular, we provide the construction
of separable Poincar\'e duality algebras, yielding a large class of
noncommutative examples. \S\ref{NCRRTheorem} reviews how the diagram
calculus is used to formulate the Grothendieck-Riemann-Roch theorem
that was proved in~\cite{BMRS} for a class of separable
$C^*$-algebras. Using this theorem, in~\S\ref{sect:dbrane} we analyse
D-branes and their charges in noncommutative spaces, explaining how to
generalize the commutative case in the context of the bivariant
theories. Several noncommutative examples are included, where the
D-brane charge formula is explicitly
calculated. In~\S\ref{CorrTduality}, we give a definition of
$\KK$-theory in terms of certain equivalence classes of noncommutative
correspondences, and relate it to T-duality in Type~II superstring
theory for noncommutative spacetime manifolds.

\subsection*{Acknowledgments}

We thank R.~Nest and R.~Reis for helpful discussions. V.M.\ was
supported by the Australian Research Council. J.R.\ was partially
supported by US NSF grant DMS-0504212. R.J.S.\ was supported in part
by the EU-RTN Network Grant MRTN-CT-2004-005104. Some of this work was
begun while J.B., V.M.\ and J.R.\ were visiting the Erwin
Schr\"odinger International Institute for Mathematical Physics as part
of the program on Gerbes, Groupoids and Quantum Field Theory.

\section{$\KK$-theory and the diagram calculus}\label{sect:KKdiagrams}

In this section we will explain the basic aspects of Kasparov's
$\KK$-theory using the novel diagram calculus developed in~\cite{BMRS}
which greatly simplifies detailed calculations in this theory. These
considerations apply in any bivariant theory with
similar algebraic properties, such as Puschnigg's local bivariant
cyclic cohomology which will be encountered later on in this section.

\subsection{Axioms}
\label{sec:KKaxioms}

For any pair of separable $C^*$-algebras $\alg$ and $\balg$, Kasparov
defines two abelian groups $\KK_0(\alg,\balg)$ and $\KK_1(\alg,
\balg)$. They are constructed with the help of Kasparov's $\alg$-$\balg$
bimodules which generalize Fredholm modules. When $\alg= \bbc$, the group
\emph{$\KK_\bullet(\bbc, \alg) = \K_\bullet(\alg)$} is the $\K$-theory 
of $\alg$, while \emph{$\KK_\bullet(\alg, \bbc) = \K^\bullet(\alg)$} is
the $\K$-homology of $\alg$. There is a canonical functor from the
category of separable $C^*$-algebras with $*$-homomorphisms to an
additive category $\underline{\KK}\,$, whose objects are separable
$C^*$-algebras and the morphisms between any two objects $\alg,\balg$
are given by ${\rm
  Mor}_{\underline{\KK}}(\alg,\balg)=\KK(\alg,\balg)$~\cite{Higson}. In 
particular, any algebra homomorphism $\phi: \alg \rar \balg$
determines an element $[\phi]_{\KK}\in \KK(\alg,\balg)$, represented
by the bimodule $(\balg,\phi,0)$. Elements of
$\KK(\alg, \balg)$ may in this way be regarded as generalized
morphisms between the two algebras. The category $\underline{\KK}$ is
not an abelian category, but it admits the structure of a triangulated
category~\cite{Nest,CMR}.

The $\KK$-theory groups may be axiomatically characterized by the
following three properties~\cite{Higson}:
\begin{enumerate}
\item \emph{Homotopy}: The bifunctor $\KK(-,-)$ is homotopy invariant
  in both variables;
\item \emph{Stability}:
  $\KK(\alg\otimes\mathcal{K},\balg)\cong\KK(\alg,\balg)\cong
  \KK(\alg,\balg\otimes\mathcal{K})$, where
  $\mathcal{K}$ is the algebra of compact operators on a separable
  Hilbert space, and the isomorphisms are induced by the stabilisation
  maps $\alg\to\alg\otimes\mathcal{K}$,
  $\balg\to\balg\otimes\mathcal{K}$ given by $a\mapsto a\otimes e$,
  $b\mapsto b\otimes e$ with $e$ a projection of rank one; and 
\item \emph{Split exactness}: If
$$
0~\longrightarrow~\mathcal{J}~\longrightarrow~\dalg~\stackrel{\scriptstyle
\longleftarrow}{\scriptstyle\longrightarrow}~\dalg/\mathcal{J}~
\longrightarrow~0
$$
is a split exact sequence of separable $C^*$-algebras and
$*$-homomorphisms, then there are split exact sequences of abelian
groups given by
\bea
0&\longrightarrow&\KK(\alg,\mathcal{J})~\longrightarrow~\KK(\alg,
\dalg)~\stackrel{\scriptstyle
\longleftarrow}{\scriptstyle\longrightarrow}~\KK(\alg,\dalg/
\mathcal{J})~\longrightarrow~0 \ , \nonumber\\[4pt]
0&\longrightarrow&\KK(\dalg/\mathcal{J},\balg)~\stackrel{\scriptstyle
\longleftarrow}{\scriptstyle\longrightarrow}~\KK(\dalg,\balg)~
\longrightarrow~\KK(\mathcal{J},\balg)~\longrightarrow~0 \ .
\nonumber
\eea
\end{enumerate}
These properties characterize $\underline{\KK}$ as a \emph{universal}
category which may be constructed by purely algebraic means. Hence it
is possible to construct the Kasparov groups as the unique formal
bifunctor on the category of separable $C^*$-algebras satisfying
the three axioms above. This point of view is extremely useful
for obtaining alternative presentations of the $\KK$-theory groups,
such as the ones we consider in~\S\ref{CorrTduality}.

\subsection{Diagrams}

The realization above of elements of $\KK(\alg,\balg)$ in terms of
generalized morphisms motivates the use of diagrams to represent classes in
$\KK$-theory which are constructed as follows. The first argument in
$\KK(-,-)$ provides the input and the second argument provides the
output of a diagram. Each tensor factor in the first argument will
produce one input node, and each tensor factor in the second argument
will produce one output node. It is sometimes convenient to add
arrowheads that point from the inputs to the outputs.

Starting with the simplest situation, an element $\alpha \in \KK(\alg,
\balg)$ is represented by the diagram 
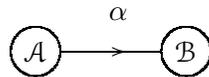
\[
\simple{\alg}{\balg}{\alpha}
\]
where the left node is regarded as the input node, while the node on
the right is the output node. The distinction between the input and
output nodes is very important in explicit computations. When
$\alg=\balg = \bbc\,$, the corresponding group is $\KK(\bbc,
\bbc)\cong\zed$.

We will also need to treat more complicated situations. For example,
an element $$\beta ~\in~ \KK(\alg\otimes \balg, \calg)$$ is
represented as
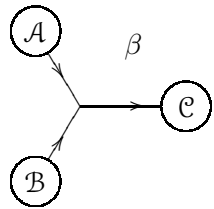
\[
\tripleright{\alg}{\balg}{\calg}{\beta}
\]
while a class $\gamma\in \KK(\alg\otimes\balg, \calg\otimes\dalg)$ is
depicted by
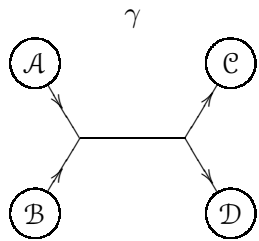
\[
\four{\alg}{\balg}{\calg}{\dalg}{\gamma}
\]
The basic rule in treating these more complicated diagrams is that
permutation of the input or output terminals may involve at most the
switch of a sign, depending on the orientation of the Bott element.

\subsection{Composition product\label{CompProd}}

A key feature of Kasparov's $\KK$-theory is the existence of the
composition (or intersection) product
\[ 
\otimes_\balg \,:\,
\KK_i(\alg, \balg) \times \KK_j(\balg, \calg) ~\longrightarrow~ 
\KK_{i+j}(\alg, \calg) \ ,
\]
which may be illustrated with the diagram
\[
\xy
(0,0)*{\alg}; 
(0,0)*\xycircle(3.3,3.3){-}="a"; 
(20,0)*{\balg}; 
(20,0)*\xycircle(3.3,3.3){-}="b"; 
"a"; "b" **\dir{-}?(.6)*\dir{>}+(-1,5)*{\alpha};
(35,0)*{\balg}; 
(35,0)*\xycircle(3.3,3.3){-}="x"; 
(55,0)*{\calg}; 
(55,0)*\xycircle(3.3,3.3){-}="y"; 
"x"; "y" **\dir{-}?(.6)*\dir{>}+(-1,5)*{\beta};
"b"; "x"**\dir{~};
\endxy
\quad = \quad 
\xy
(0,0)*{\alg}; 
(0,0)*\xycircle(3.3,3.3){-}="a"; 
(20,0)*{\calg}; 
(20,0)*\xycircle(3.3,3.3){-}="b"; 
"a"; "b" **\dir{-}?(.6)*\dir{>}+(-1,5)*{\alpha\otimes_\balg\beta};
\endxy
\]
Here $\alpha \in \KK_i(\alg, \balg)$, $\beta\in \KK_j(\balg, \calg)$
and $\alpha\otimes _\balg\beta\in \KK_{i+j}(\alg, \calg)$. The wavy
line joining the output node of $\alpha$ to the input node of
$\beta$ serves to stress that these two nodes, which have the same
label $\balg$, are annihilated in the process of taking the composition
product. In particular, if $\phi: \alg\rar \balg$ and $\psi: \balg\rar
\calg$ are morphisms of $C^*$-algebras, then this product is
compatible with composition of morphisms, giving
$$
[\phi]_{\KK}\otimes_\balg[\psi]_{\KK} = [\psi\circ \phi]_{\KK} \ .
$$

Kasparov's composition product is bilinear and associative, and it
defines the composition law in the additive category
$\underline{\KK}$~\cite{Higson}. It also makes
$\KK_0(\alg, \alg)$ into a unital ring whose unit $1_\alg$ is the
element of $\KK_0(\alg, \alg)$ determined by the identity morphism
$\Id_\alg: \alg \rar\alg$ of the algebra $\alg$. In particular, it now
makes sense to talk about invertible elements of $\KK_0(\alg,
\alg)$. More generally, we say that $\alpha \in \KK_i(\alg,\balg)$ is
\emph{invertible} if and only if there exists an element $\beta \in
\KK_{-i}(\balg,\alg)$ such that $\alpha\otimes _\balg\beta = 1_\alg$
and $\beta\otimes_\alg\alpha=1_\balg$. The element $\beta$ is then
called the \emph{inverse} of $\alpha$ and written $\beta=\alpha^{-1}$.

Elements of $\KK$-theory determine, by means of the composition
product, homomorphisms in $\K$-theory and $\K$-homology in the
following way. Let $ \alpha \in \KK_d(\alg, \balg)$. Then $$x\otimes
_\alg\alpha  ~\in~ \KK_{j+d}(\bbc,\balg) = \K_{j+d}(\balg)$$ for any $x\in
\KK_j(\bbc, \alg) = \K_j(\alg)$. Thus taking the product with $\alpha$
on the right gives a map $(-)\otimes _\alg \alpha: \KK_j(\bbc, \alg)
\rightarrow \KK_{j+d} (\bbc, \balg)$, which in diagrammatic form is
\[
\xy
(0,0)*{\bbc}; 
(0,0)*\xycircle(3.3,3.3){-}="a"; 
(20,0)*{\alg}; 
(20,0)*\xycircle(3.3,3.3){-}="b"; 
"a"; "b" **\dir{-}?(.6)*\dir{>}+(-1,5)*{x};
(35,0)*{\alg}; 
(35,0)*\xycircle(3.3,3.3){-}="x"; 
(55,0)*{\balg}; 
(55,0)*\xycircle(3.3,3.3){-}="y"; 
"x"; "y" **\dir{-}?(.6)*\dir{>}+(-1,5)*{\alpha};
"b"; "x"**\dir{~};
\endxy
\quad = \quad 
\xy
(0,0)*{\bbc}; 
(0,0)*\xycircle(3.3,3.3){-}="a"; 
(20,0)*{\balg}; 
(20,0)*\xycircle(3.3,3.3){-}="b"; 
"a"; "b" **\dir{-}?(.6)*\dir{>}+(-1,5)*{x\otimes_\alg\alpha};
\endxy
\]
Similarly, $\alpha \otimes_\balg y \in \KK_{j+d}(\alg, \bbc)=\K^{j+d}(\alg)$
for all $y\in \KK_j(\balg, \bbc)=\K^j(\balg)$, and in this way
multiplying by the element  $\alpha$ on the left gives a map
$\alpha\otimes_\balg(-): \KK_j(\balg,\bbc) \rightarrow \KK_{j+d}(\alg,
\bbc)$, which in diagrams is
\[
\xy
(0,0)*{\alg}; 
(0,0)*\xycircle(3.3,3.3){-}="a"; 
(20,0)*{\balg}; 
(20,0)*\xycircle(3.3,3.3){-}="b"; 
"a"; "b" **\dir{-}?(.6)*\dir{>}+(-1,5)*{\alpha};
(35,0)*{\balg}; 
(35,0)*\xycircle(3.3,3.3){-}="x"; 
(55,0)*{\bbc}; 
(55,0)*\xycircle(3.3,3.3){-}="y"; 
"x"; "y" **\dir{-}?(.6)*\dir{>}+(-1,5)*{y};
"b"; "x"**\dir{~};
\endxy
\quad = \quad 
\xy
(0,0)*{\alg}; 
(0,0)*\xycircle(3.3,3.3){-}="a"; 
(20,0)*{\bbc}; 
(20,0)*\xycircle(3.3,3.3){-}="b"; 
"a"; "b" **\dir{-}?(.6)*\dir{>}+(-1,5)*{\alpha\otimes_\balg y};
\endxy
\]
An invertible element $\alpha \in \KK(\alg, \balg)$
determines an isomorphism between the $\K$-theory of $\alg$ and the
$\K$-theory of $\balg$, as well as an isomorphism between the
$\K$-homology of $\alg$ and the $\K$-homology of $\balg$. In this case
we say that $\alg$ and $\balg$ are (\emph{strongly})
\emph{$\KK$-equivalent}.

\subsection{Exterior product}

A very important computational tool in $\KK$-theory is the bilinear
exterior product
\[
\otimes \,:\,
\KK_i(\alg_1, \balg_1) \times \KK_j(\alg_2, \balg_2) ~
\longrightarrow~ \KK_{i+j}(\alg_1\otimes \alg_2, \balg_1\otimes 
\balg_2)
\]
which is defined in terms of the composition product by 
\[
x_1\otimes x_2 = (x_1 \otimes 1_{\alg_2} )\otimes_{\balg_1\otimes
  \alg_2} (x_2\otimes 1_{\balg_1}) \ .
\]
The exterior product with the class of the identity morphism, called
\emph{dilation}, is defined as follows. If $x \in \KK_j(\alg, \balg)$
is represented by a Kasparov $\alg$-$\balg$ bimodule
$(\mathcal{E}, F)$, then $x\otimes 1_\calg$ is the element of
$\KK_j(\alg\otimes \calg, \balg\otimes \calg)$ represented by the
$(\alg\otimes \calg)$-$( \balg\otimes \calg)$ bimodule
$(\mathcal{E}\otimes \calg, F\otimes \Id_\calg)$. In
the simple case where $x_1 = [\phi]_{\KK}$ is the class of an algebra
morphism $\phi: \alg_1 \rightarrow \balg_1$, then $x_1\otimes
1_{\alg_2} = [\phi\otimes\Id_{\alg_2}]_{\KK}$ is the class of the morphism
$\phi\otimes\Id_{\alg_2}: \alg_1\otimes \alg_2 \rightarrow
\balg_1\otimes \alg_2$. Diagrammatically, dilation is illustrated as
\[
\simple{\alg_1}{\balg_1}{x_1}\qquad
\longmapsto\qquad
\four{\alg_1}{\alg_2}{\balg_1}{\alg_2}{x_1\otimes 1_{\alg_2}}
\]
When taking the exterior product of elements $x_1 \in \KK(\alg_1,
\balg_1)$ and $x_2\in \KK(\alg_2, \balg_2)$ we must use dilation to
establish a common output-input node to which we can apply the
Kasparov composition product. This is described in the diagram
\begin{equation}\label{dilation}
\begin{split}
\simple{\alg_1}{\balg_1}{x_1} & \quad \otimes \quad
\simple{\alg_2}{\balg_2}{x_2}\\ \\
& = \quad
\xy
(0,10)*{\alg_1}; 
(0,10)*\xycircle(3.3,3.3){-}="a"; 
(0,-10)*{\alg_2}; 
(0,-10)*\xycircle(3.3,3.3){-}="b"; 
(6,0)*{}="d";
 "a"; "d" **\dir{-}?(.6)*\dir{>};
  "b";"d" **\dir{-}?(.6)*\dir{>};
  (20,0)*{}="c";
  (26,10)*{\balg_1}; 
(26,10)*\xycircle(3.3,3.3){-}="x"; 
(26,-10)*{\alg_2}; 
(26,-10)*\xycircle(3.3,3.3){-}="y"; 
"c" ; "x" **\dir{-}?(.6)*\dir{>};
"c"; "y" **\dir{-}?(.6)*\dir{>};
"d"; "c" **\dir{-}; 
(13,16)*{x_1\otimes 1_{\alg_2}};
(40,10)*{\balg_1}; 
(40,10)*\xycircle(3.3,3.3){-}="a1"; 
(40,-10)*{\alg_2}; 
(40,-10)*\xycircle(3.3,3.3){-}="b1"; 
(46,0)*{}="d1";
 "a1"; "d1" **\dir{-}?(.6)*\dir{>};
  "b1";"d1" **\dir{-}?(.6)*\dir{>};
  (60,0)*{}="c1";
  (66,10)*{\balg_1}; 
(66,10)*\xycircle(3.3,3.3){-}="x1"; 
(66,-10)*{\balg_2}; 
(66,-10)*\xycircle(3.3,3.3){-}="y1"; 
"c1" ; "x1" **\dir{-}?(.6)*\dir{>};
"c1"; "y1" **\dir{-}?(.6)*\dir{>};
"d1"; "c1" **\dir{-}; 
(53,16)*{1_{\balg_1}\otimes x_2};
(29,10)*{}="f"; 
(37,10)*{}="g"; 
"f"; "g" **\dir{~}; 
(29,-10)*{}="f1"; 
(37,-10)*{}="g1"; 
"f1" ; "g1" **\dir{~}; 
\endxy
\\
 \\
&= \quad \four{\alg_1}{\alg_2}{\balg_1}{\balg_2}{(x_1\otimes 1_{\alg_2})
\otimes_{\balg_1\otimes\alg_2}
(1_{\balg_1}\otimes x_2)}
\end{split}
\end{equation} 

The diagram calculus allows one to keep track of the exterior product
in more complicated situations. For example, if $\alg_i,\balg_i$,
$i=1,2$ and $\dalg$ are all separable $C^*$-algebras, then the
exterior product $x_1\otimes _\dalg x_2$ of two elements $x_1\in
\KK(\alg_1, \balg_1\otimes \dalg)$ and $x_2\in \KK(\dalg\otimes
\alg_2, \balg_2)$ can be illustrated as
\begin{equation}
\xy
(0,0)*{\alg_1}; 
(0,0)*\xycircle(3.3,3.3){-}="a"; 
(20,10)*{\balg_1}; 
(20,10)*\xycircle(3.3,3.3){-}="b"; 
(20,-10)*{\dalg}; 
(20,-10)*\xycircle(3.3,3.3){-}="c"; 
(14,0)*{}="d";
"a"; "d" **\dir{-}?(.6)*\dir{>};
"d"; "b" **\dir{-}?(.6)*\dir{>};
"d"; "c" **\dir{-}?(.6)*\dir{>};
(7,8)*{x_1}; 
(35,10)*{\alg_2}; 
(35,10)*\xycircle(3.3,3.3){-}="x"; 
(35,-10)*{\dalg}; 
(35,-10)*\xycircle(3.3,3.3){-}="y"; 
(55,0)*{\balg_2}; 
(55,0)*\xycircle(3.3,3.3){-}="z"; 
(41,0)*{}="t";
 "x";"t" **\dir{-}?(.6)*\dir{>};
 "t"; "z"**\dir{-}?(.6)*\dir{>};
 "y"; "t" **\dir{-}?(.6)*\dir{>};
(48,8)*{x_2}; 
"c" ; "y" **\dir{~};
\endxy
\label{extprodDdiag}\end{equation}
While this product is defined by means of the dilation method as
before, the diagrammatic calculus simplifies the process. To obtain
the diagram describing the result, we first remove the common nodes
joined by the wavy line, and then assemble input and output nodes
together. In the diagram~(\ref{extprodDdiag}), the input nodes are
labelled $\alg_1$ and $\alg_2$ while the output nodes are $\balg_1$
and $\balg_2$. Thus the result looks like this:
\[
\four{\alg_1}{\alg_2}{\balg_1}{\balg_2}{x_1\otimes _\dalg x_2} ~\in~ 
\KK(\alg_1\otimes \alg_2, 
\balg_1\otimes \balg_2) \ .
\]
When $\dalg = \bbc$ this calculation is the same as that in
eq.~(\ref{dilation}). 

\subsection{Associativity}

The exterior product is associative, and this highly non-trivial
property may be stated as follows. We assume that all algebras
involved below are separable.

\begin{theorem}\label{associative}
Let $x\in \KK(\alg_1, \balg_1\otimes \dalg_1)$, $y\in
\KK(\dalg_1\otimes \alg_2, \balg_2\otimes\dalg_2)$ and
$z\in\KK(\dalg_2\otimes \alg_3, \balg_3)$. Then
$$
(x\otimes_{\dalg_1} y)\otimes_{\dalg_2}z =
x\otimes_{\dalg_1}(y\otimes_{\dalg_2} z) \ .
$$
\end{theorem}

The nice aspect of the diagrammatic formalism is that one can show
that all associativity formulae in Kasparov theory correspond to the
fact that one can perform the concatenations that form part of
computing the products in any order, except perhaps for signs (which
disappear in $\KK_0$). To illustrate this point, consider the
following example. Let us take three elements $z\in \KK(\ealg,
\balg)$, $y\in\KK(\dalg, \alg)$ and $x\in \KK(\alg\otimes\balg,
\calg)$. We would like to compute the product
$z\otimes_\balg(y\otimes_\alg x)$. First we compute $y\otimes_\alg x$
to get
\begin{displaymath}
y\otimes_\alg x  ~= \quad
\xy
(0,-10)*{\dalg}; 
(0,-10)*\xycircle(3.3,3.3){-}="a"; 
(20,-10)*{\alg}; 
(20,-10)*\xycircle(3.3,3.3){-}="b"; 
"a"; "b" **\dir{-}?(.6)*\dir{>}+(-1,5)*{y};
(35,10)*{\balg}; 
(35,10)*\xycircle(3.3,3.3){-}="x"; 
(35,-10)*{\alg}; 
(35,-10)*\xycircle(3.3,3.3){-}="y"; 
(55,0)*{\calg}; 
(55,0)*\xycircle(3.3,3.3){-}="z"; 
(41,0)*{}="t";
 "x";"t" **\dir{-}?(.6)*\dir{>};
 "t"; "z"**\dir{-}?(.6)*\dir{>};
 "y"; "t" **\dir{-}?(.6)*\dir{>};
(48,8)*{x}; 
"b" ; "y"** \dir{~}
\endxy  \quad = \quad \tripleright{\balg}{\dalg}{\calg}{y\otimes_\alg
  x}
\end{displaymath} 
so that
\begin{displaymath}
z\otimes_\balg(y\otimes_\alg x)   ~= \quad
\xy
(0,10)*{\ealg}; 
(0,10)*\xycircle(3.3,3.3){-}="a"; 
(20,10)*{\balg}; 
(20,10)*\xycircle(3.3,3.3){-}="b"; 
"a"; "b" **\dir{-}?(.6)*\dir{>}+(-1,5)*{z};
(35,10)*{\balg}; 
(35,10)*\xycircle(3.3,3.3){-}="x"; 
(35,-10)*{\dalg}; 
(35,-10)*\xycircle(3.3,3.3){-}="y"; 
(55,0)*{\calg}; 
(55,0)*\xycircle(3.3,3.3){-}="z"; 
(41,0)*{}="t";
 "x";"t" **\dir{-}?(.6)*\dir{>};
 "t"; "z"**\dir{-}?(.6)*\dir{>};
 "y"; "t" **\dir{-}?(.6)*\dir{>};
(48,8)*{y\otimes_\alg x}; 
"b" ; "x"** \dir{~}
\endxy \quad = \quad
\tripleright{\ealg}{\dalg}{\calg}{\qquad z\otimes_\balg(y\otimes_\alg x)}
\end{displaymath}
However, since the concatenations involved in this calculation can be
performed in any order, the same computation may be carried out as
\[
\xy
(0,-10)*{\dalg}; 
(0,-10)*\xycircle(3.3,3.3){-}="a"; 
(20,-10)*{\alg}; 
(20,-10)*\xycircle(3.3,3.3){-}="b"; 
"a"; "b" **\dir{-}?(.6)*\dir{>}+(-1,5)*{y};
(15,10)*{\ealg};
(15,10)*\xycircle(3.3,3.3){-}="xx"; 
(35,10)*{\otimes_\balg}; 
(35,10)*\xycircle(3.3,3.3){.}="x"; 
(35,-10)*{\alg}; 
(35,-10)*\xycircle(3.3,3.3){-}="y"; 
(55,0)*{\calg}; 
(55,0)*\xycircle(3.3,3.3){-}="z"; 
(41,0)*{}="t";
 "x";"t" **\dir{-}?(.6)*\dir{>};
 "t"; "z"**\dir{-}?(.6)*\dir{>};
 "y"; "t" **\dir{-}?(.6)*\dir{>};
(48,8)*{x}; 
"xx"; "x"**\dir{-}?(.6)*\dir{>}+(-1,5)*{z};
"b" ; "y"** \dir{~}
\endxy  \quad
 = \quad
\tripleright{\ealg}{\dalg}{\calg}{\qquad y\otimes_\alg(z\otimes_\balg
  x)} 
\]
where the dashed circle indicates a node that is removed in the
process of computing the Kasparov product with respect to the algebra
marked in the circle. The two results are the same, which leads to the
formula
\[
z\otimes_\balg(y\otimes_\alg x) = \pm \,y\otimes _\alg(z\otimes_\balg
x) \ .
\]
A rigorous proof of this formula can be found
in~\cite[Appendix~B]{BMRS}.

\subsection{Local bivariant cyclic cohomology\label{LBCC}}

To make the Chern character possible, one needs a target group for a
map defined on Kasparov's bivariant $\KK$-theory, preferably one that
is defined on a similar class of algebras and which
transports the key algebraic properties of the theory. Several theories
of such kind have been proposed, which work for large classes of
topological and bornological algebras. For our purposes, the most
convenient theory is Puschnigg's local bivariant cyclic
cohomology~\cite{Puschnigg}, which we shall denote $\HE$. This theory,
which can be defined for separable $C^*$-algebras, has the same formal
properties as Kasparov's $\KK$-theory. In particular, if all algebras
below are separable $C^*$-algebras then
\begin{enumerate}
\item There exists a bilinear, associative composition product 
$$
\otimes_\balg \,:\,
\HE_i(\alg,\balg)\times \HE_j(\balg,\calg) ~
\longrightarrow~ \HE_{i+j}(\alg,\calg) \ ;
$$
\item The bifunctor $\HE(-,-)$ is homotopy invariant, split exact and
  satisfies excision in each variable; and
\item
There exists an associative, bilinear exterior product
\[
\otimes \,:\,\HE_i(\alg_1,\balg_1)\times \HE_j(\alg_2,\balg_2) ~
\longrightarrow ~
\HE_{i+j}(\alg_1 \otimes \alg_2,\balg_1 \otimes \balg_2) \ .
\]
\end{enumerate}
Thus essentially everything in our description of $\KK$-theory above
applies also to the bivariant cyclic homology.

This theory is also defined for more general topological algebras,
such as smooth dense subalgebras of $C^*$-algebras. 
\begin{theorem} \cite[\S23.4]{CuntzSurvey}. 
\label{thm:approxprop}
Let $\alg$ be a Banach algebra with the metric approximation
property. Let $\alg^\infty$ be a smooth subalgebra of $\alg$. Then the
inclusion map $\alg^\infty\hookrightarrow \alg$ induces an invertible
element of $\HE_0(\alg^\infty,\alg)$, hence $\alg^\infty$ and $\alg$ are
$\HE$-equivalent.
\end{theorem}

All nuclear $C^*$-algebras have the metric approximation
property~\cite{ChoiE}.  The example of the Fr\'echet algebra 
$C^\infty(X)$ of smooth functions on a compact manifold $X$ is of key
importance. This algebra is a smooth subalgebra of the $C^*$-algebra of
continuous functions $C(X)$. By Theorem~\ref{thm:approxprop} the
inclusion $C^\infty(X) \hookrightarrow C(X)$ is an invertible element
in $\HE(C^\infty(X),C(X))$. In particular, both the local cyclic homology and
cohomology of these two algebras are isomorphic. Puschnigg proves that
$\HE_\bullet(C^\infty(X))\cong\HP_\bullet(C^\infty(X))$, so that the local 
cyclic homology coincides with the standard periodic cyclic
homology. Combined with a fundamental result of Connes, this
establishes an isomorphism between Puschnigg's local cyclic homology
of $C(X)$ and the periodic de~Rham cohomology of $X$,
$\HE_\bullet(C(X))\cong\H^\bullet(\Omega^\#(X),\dd)$ (with the
standard grading on differential forms by form degree).

The local bivariant cyclic homology admits a Chern character with the
required properties.

\begin{theorem} \cite[\S23.5]{CuntzSurvey} Let $\alg$ and
  $\balg$ be separable $C^*$-algebras. Then  
there exists a natural bivariant Chern character 
$$
\ch \,:\, \KK_\bullet(\alg,\balg) ~
\longrightarrow~ \HE_\bullet(\alg,\balg)
$$
such that
\begin{enumerate}
\item $\ch$ is multiplicative, i.e., if
$\alpha\in \K\K_i(\alg,\balg)$ and $\beta\in \K\K_j(\balg,\calg)$ then
$$
\ch(\alpha\otimes_\balg\beta) = \ch (\alpha) \otimes_\balg \ch(\beta)
\ ;
$$
\item $\ch$ is compatible with the exterior product; and
\item $\ch([\phi]_{\KK})=[\phi]_\HE$ for any algebra homomorphism
  $\phi:\alg\to\balg$.
\end{enumerate}
\label{ChernHLthm}\end{theorem}

The last property implies that the Chern character sends invertible
elements of $\K\K$-theory to invertible elements of bivariant local
cyclic cohomology. In particular, $\KK$-equivalence of algebras
implies $\HE$-equivalence, but not conversely. Moreover, if $\alg$ and
$\balg$ are in the class $\underline{\mathfrak N}$ of $C^*$-algebras
for which the universal coefficient theorem holds in $\KK$-theory,
\emph{i.e.}, those which are $\KK$-equivalent to commutative
$C^*$-algebras, then
\beq
\HE_\bullet(\alg,\balg)
~\cong~\Hom_\bbC\big(\K_\bullet(\alg)\otimes_{\bbz}\complex \,,\,
\K_\bullet(\balg)\otimes_\bbz\complex\big) \ .  
\label{HLHomiso}\eeq
If the $\K$-theory $\K_\bullet(\alg)$ is finitely generated, this is
also equal to $\KK_\bullet(\alg,\balg) \otimes_\bbz\complex$.

\section{Noncommutative Poincar\'e duality}\label{sect:ncPD}

Following Connes~\cite{ConnesBook}, in this section we will introduce
the notion of Poincar\'e duality suitable to generic separable
$C^*$-algebras.

\subsection{Fundamental classes}

Given an algebra $\alg$, let $\alg^\op$ denote the opposite algebra of
$\alg$, \emph{i.e.}, the algebra with the same underlying vector space
as $\alg$ but with the product reversed. The stable homotopy category
of {\Ca}s has an involution which is defined by the mapping
$$
\big(\alg\xrightarrow{f}\balg\big)~\longmapsto~\big(
\alg^\op\xrightarrow{f^\op}\balg^\op\big) \ ,
$$
and this involution passes to the category $\underline{\KK}$. Thus given
$x\in \KK_d(\alg,\balg)$, there is a corresponding element $x^\op \in
\KK_d(\alg^\op,\balg^\op)$ with $(x^\op)^\op = x$.

\begin{definition} 
We say that a separable $C^*$-algebra $\alg$ is a (\emph{strong})
  \emph{Poincar\'{e} duality} (\emph{PD}) \emph{algebra} if and only
  if there exists a \emph{fundamental class} for $\alg$ in
  $\K$-homology, \emph{i.e.,} an element $\Delta\in\KK_d(\alg\otimes\alg^\op,
\bbc)=\K^d(\alg\otimes\alg^\op)$ with an \emph{inverse}
class $\Delta^\vee \in \KK_{-d}(\bbc, \alg\otimes
\alg^\op)=\K_{-d}(\alg\otimes\alg^\op)$ such that 
\bea
\Delta^\vee \otimes _{\alg^\op}\Delta &=& 1_\alg \in \KK_0(\alg, \alg)
\ , \nonumber\\[4pt]
\Delta^\vee \otimes _{\alg}\Delta &=& (-1)^d\, 1_{\alg^\op} \in
\KK_0(\alg^\op, \alg^\op) \ .
\label{Poincareconds}\eea
\end{definition}

The use of the opposite algebra in this definition is to describe
$\alg$-bimodules as $(\alg\otimes\alg^\op)$-modules. The classes
$\Delta$ and $\Delta^\vee$ can be illustrated with the help of the
diagrams
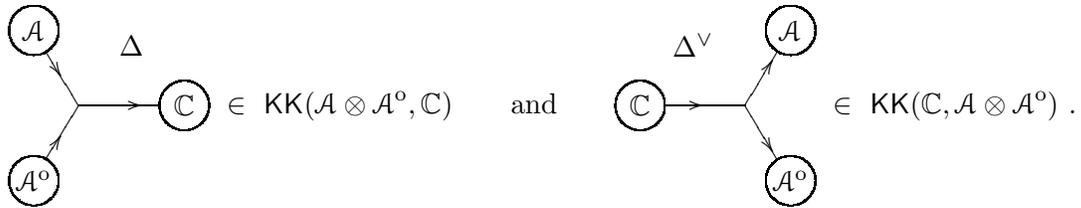
\[
\tripleright{\alg}{\alg^\op}{\bbc}{\Delta} ~\in~ 
\KK(\alg\otimes \alg^\op, \bbc) \qquad \mbox{and} \qquad
\tripleleft{\bbc}{\alg}{\alg^\op}{\Delta^\vee}~\in ~
\KK(\bbc, \alg\otimes\alg^\op) \ .
\]
The Poincar\'{e} conditions (\ref{Poincareconds}) can then be
illustrated as follows. The first condition gives rise to the
diagram
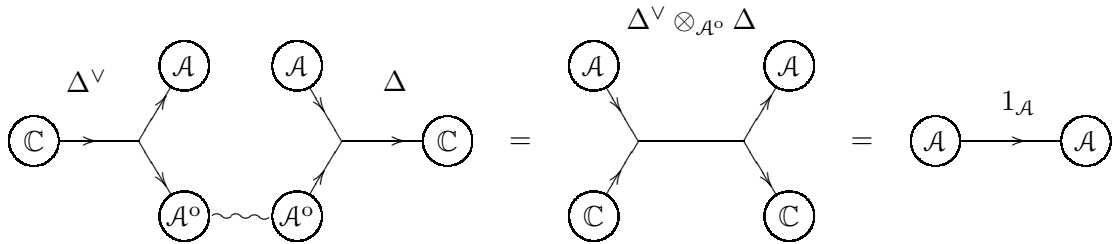
\[
\xy
(0,0)*{\bbc}; 
(0,0)*\xycircle(3.3,3.3){-}="a"; 
(20,10)*{\alg}; 
(20,10)*\xycircle(3.3,3.3){-}="b"; 
(20,-10)*{\alg^\op}; 
(20,-10)*\xycircle(3.3,3.3){-}="c"; 
(14,0)*{}="d";
"a"; "d" **\dir{-}?(.6)*\dir{>};
"d"; "b" **\dir{-}?(.6)*\dir{>};
"d"; "c" **\dir{-}?(.6)*\dir{>};
(7,8)*{\Delta^\vee}; 
(35,10)*{\alg}; 
(35,10)*\xycircle(3.3,3.3){-}="x"; 
(35,-10)*{\alg^\op}; 
(35,-10)*\xycircle(3.3,3.3){-}="y"; 
(55,0)*{\bbc}; 
(55,0)*\xycircle(3.3,3.3){-}="z"; 
(41,0)*{}="t";
 "x";"t" **\dir{-}?(.6)*\dir{>};
 "t"; "z"**\dir{-}?(.6)*\dir{>};
 "y"; "t" **\dir{-}?(.6)*\dir{>};
(48,8)*{\Delta}; 
"c" ; "y" **\dir{~};
\endxy
\quad = \quad \four{\alg}{\bbc}{\alg}{\bbc}{\Delta^\vee
\otimes_{\alg^\op}\Delta} \quad
= \quad \simple{\alg}{\alg}{1_\alg}
\]
which represents the identity element $1_\alg \in \KK_0(\alg,
\alg)$. The second condition (in the case $d$ even)
is described by the diagram
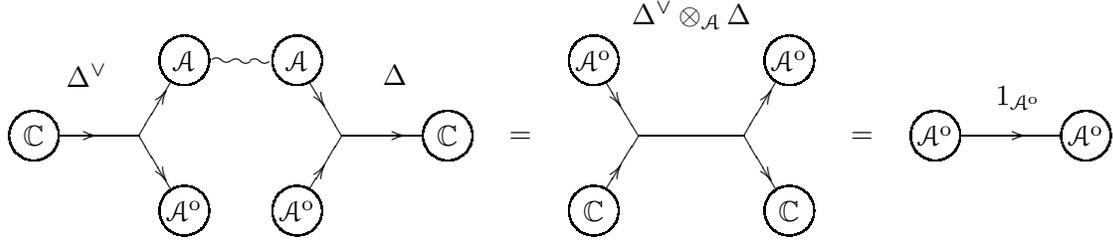
\[
\xy
(0,0)*{\bbc}; 
(0,0)*\xycircle(3.3,3.3){-}="a"; 
(20,10)*{\alg}; 
(20,10)*\xycircle(3.3,3.3){-}="b"; 
(20,-10)*{\alg^\op}; 
(20,-10)*\xycircle(3.3,3.3){-}="c"; 
(14,0)*{}="d";
"a"; "d" **\dir{-}?(.6)*\dir{>};
"d"; "b" **\dir{-}?(.6)*\dir{>};
"d"; "c" **\dir{-}?(.6)*\dir{>};
(7,8)*{\Delta^\vee}; 
(35,10)*{\alg}; 
(35,10)*\xycircle(3.3,3.3){-}="x"; 
(35,-10)*{\alg^\op}; 
(35,-10)*\xycircle(3.3,3.3){-}="y"; 
(55,0)*{\bbc}; 
(55,0)*\xycircle(3.3,3.3){-}="z"; 
(41,0)*{}="t";
 "x";"t" **\dir{-}?(.6)*\dir{>};
 "t"; "z"**\dir{-}?(.6)*\dir{>};
 "y"; "t" **\dir{-}?(.6)*\dir{>};
(48,8)*{\Delta}; 
"b" ; "x" **\dir{~};
\endxy
\quad = \quad \four{\alg^\op}{\bbc}{\alg^\op}{\bbc}{\Delta^\vee
\otimes_{\alg}\Delta}  \quad
= \quad \simple{\alg^\op}{\alg^\op}{1_{\alg^\op}}
\]

\begin{definition}
\label{thm:symmetry}
A fundamental class $\Delta$ of a PD algebra $\alg$ is said to be
\emph{symmetric} if $\sigma(\Delta)^\op = \Delta $ in $ \K^d(\alg\otimes 
\alg^\op)$, where
\[
\sigma \,:\,
\alg\otimes\alg^\op ~\longrightarrow~ \alg^\op \otimes \alg
\]
is the flip involution $x\otimes y^\op \mapsto y^\op \otimes x$
and $\sigma$ also denotes the induced map on $\K$-homology.
In terms of the diagram calculus, $\Delta$ being symmetric
implies that
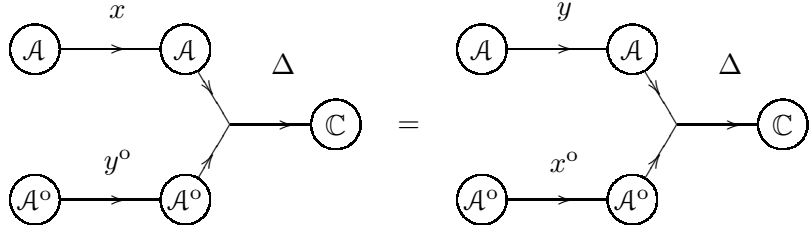
\[
\xy
(15,-10)*{\alg^\op}; 
(15,-10)*\xycircle(3.3,3.3){-}="a"; 
(35,-10)*{\alg^\op}; 
(35,-10)*\xycircle(3.3,3.3){-}="b"; 
"a"; "b" **\dir{-}?(.6)*\dir{>}+(-1,5)*{y^\op};
(15,10)*{\alg};
(15,10)*\xycircle(3.3,3.3){-}="xx"; 
(35,10)*{\alg}; 
(35,10)*\xycircle(3.3,3.3){-}="x"; 
(55,0)*{\bbc}; 
(55,0)*\xycircle(3.3,3.3){-}="z"; 
(41,0)*{}="t";
 "x";"t" **\dir{-}?(.6)*\dir{>};
 "t"; "z"**\dir{-}?(.6)*\dir{>};
 "b"; "t" **\dir{-}?(.6)*\dir{>};
(48,8)*{\Delta}; 
"xx"; "x"**\dir{-}?(.6)*\dir{>}+(-1,5)*{x};
\endxy  \quad
 = \quad
\xy
(15,-10)*{\alg^\op}; 
(15,-10)*\xycircle(3.3,3.3){-}="a"; 
(35,-10)*{\alg^\op}; 
(35,-10)*\xycircle(3.3,3.3){-}="b"; 
"a"; "b" **\dir{-}?(.6)*\dir{>}+(-1,5)*{x^\op};
(15,10)*{\alg};
(15,10)*\xycircle(3.3,3.3){-}="xx"; 
(35,10)*{\alg}; 
(35,10)*\xycircle(3.3,3.3){-}="x"; 
(55,0)*{\bbc}; 
(55,0)*\xycircle(3.3,3.3){-}="z"; 
(41,0)*{}="t";
 "x";"t" **\dir{-}?(.6)*\dir{>};
 "t"; "z"**\dir{-}?(.6)*\dir{>};
 "b"; "t" **\dir{-}?(.6)*\dir{>};
(48,8)*{\Delta}; 
"xx"; "x"**\dir{-}?(.6)*\dir{>}+(-1,5)*{y};
\endxy
\]
for all $x,y\in\KK(\alg, \alg)$.
\end{definition}

A more general form of Poincar\'{e} duality gives the notion of
Poincar\'{e} dual pairs of algebras.
\begin{definition}
A pair of separable $C^*$-algebras $\alg$ and $\balg$ is said to be  a
(\emph{strong}) \emph{Poincar\'{e} dual} (\emph{PD}) \emph{pair} if
and only if there exists a fundamental class $\Delta \in\KK_d(\alg\otimes\balg,
\bbc)=\K^d(\alg\otimes\balg)$ and an inverse class $\Delta^\vee \in
\KK_{-d}(\bbc, \alg\otimes \balg)=\K_{-d}(\alg\otimes\balg)$ such that 
\beq
\Delta^\vee \otimes _{\balg}\Delta &=& 1_\alg \in \KK_0(\alg, \alg) \
, \nonumber\\[4pt]
\Delta^\vee \otimes _{\alg}\Delta &=& (-1)^d\, 1_\balg \in
\KK_0(\balg, \balg) \  . \nonumber
\eea
\end{definition}
In analogy to Poincar\'{e} duality, we can illustrate the classes
$\Delta$ and $\Delta^\vee$ with the help of the diagrams
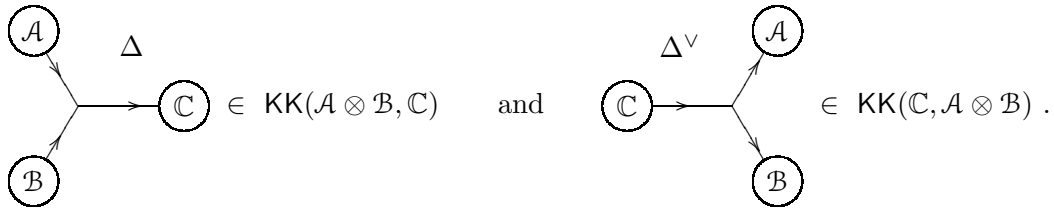
\[
\tripleright{\alg}{\balg}{\bbc}{\Delta} ~\in~ \KK(\alg\otimes\balg,
\bbc) \qquad \mbox{and} \qquad
\tripleleft{\bbc}{\alg}{\balg}{\Delta^\vee} ~\in~ 
\KK(\bbc, \alg\otimes\balg) \ .
\]
The product on the right with $\Delta$ and the product on the 
left with $\Delta^\vee$ give inverse isomorphisms
\[
\K_i(\alg) ~\xrightarrow{(-)\otimes_\alg \Delta}~ \K^{i+d}(\balg)
\qquad \mbox{and} \qquad
\K^i(\balg) ~\xrightarrow{\Delta^\vee\otimes_\balg(-)}~ 
\K_{i-d}(\alg) \ .
\]
As can be checked with the diagram calculus, one also gets Poincar\'e
duality with coefficients in any auxiliary algebras $\calg,\dalg$
given by
\[
\KK_i(\calg, \alg\otimes\dalg) ~\cong ~
\KK_{i-d}(\calg\otimes\balg, \dalg) \ .
\]
$\KK$-equivalent algebras share the same duality
properties~\cite{BMRS}.

\begin{remark}

The word \lq strong' in the definition of Poincar\'{e} duality above
is sometimes used because there are weaker notions of the duality. The
weaker conditions are described in detail in~\cite[\S2.7]{BMRS}.
\end{remark}

All of these notions carry over easily to the local
bivariant cyclic cohomology described in~\S\ref{LBCC}. Algebras
$\alg$ and $\balg$ are a (\emph{strong}) \emph{cyclic Poincar\'{e} dual}
  (\emph{C-PD}) \emph{pair} if there exists a class $$\Xi~\in~
  \HE_d(\alg\otimes\balg, \bbc)
= \HE^d(\alg\otimes\balg)$$  and a class $\Xi^\vee\in
\HE_d(\bbc,\alg\otimes\balg)=\HE_d(\alg\otimes\balg)$ such that 
\bea
 \Xi^\vee\otimes _\balg\Xi&=&
1_\alg\in \HE_0(\alg,\alg) \ , \nonumber\\[4pt]
\Xi^\vee\otimes _\alg\Xi&=&(-1)^d\,
1_\balg\in \HE_0(\balg,\balg) \ . \nonumber
\eea
The element $\Xi$ is called a \emph{fundamental cyclic cohomology
  class} for the pair $(\alg,\balg)$, and $\Xi^\vee$ is called its {\it
  inverse}. An algebra $\alg$ is a (\emph{strong}) \emph{cyclic Poincar\'{e}
  duality} (\emph{C-PD}) \emph{algebra} if $(\alg,\alg^\op)$ is a C-PD pair. By
multiplicativity of the Chern character, each PD pair can also be made
into a C-PD pair for $\HE$. However, in many instances one does
not wish to take the $\HE$ fundamental class $\Xi$ to be equal to
$\ch(\Delta)$. We will see this explicitly later on.

\subsection{Constructing PD algebras}

An application of the diagram calculus answers the question of how big
is the space of fundamental classes of an algebra.

\begin{proposition}
Let $(\alg,\balg)$ be a PD pair, and let $\Delta \in
\K^d(\alg\otimes\balg)$ be a fundamental class with inverse
$\Delta^\vee  \in \K_{-d}(\alg\otimes\balg)$. Let $\ell \in
\KK_0(\alg, \alg)$ be an invertible element. Then
$\ell\otimes_\alg\Delta \in \K^d(\alg\otimes\balg)$ is another
fundamental class, with inverse $\Delta^\vee \otimes_\alg\ell^{-1}
\in \K_{-d}(\alg\otimes\balg)$.
\end{proposition}

\begin{proof}
The harder direction to prove is that the product $(\Delta^\vee
\otimes_\alg\ell^{-1})\otimes_{\balg}(\ell\otimes_\alg\Delta)$
produces the identity element in $\KK(\alg,\alg)$. This product is
illustrated by the diagram
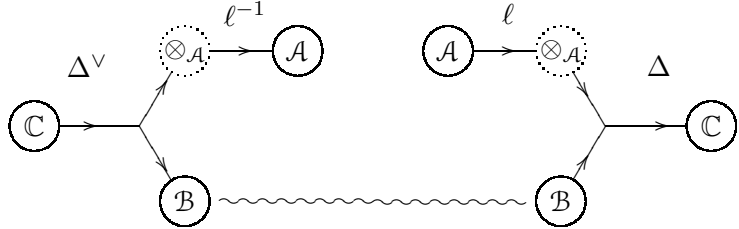
\[
\xy
(0,0)*{\bbc}; 
(0,0)*\xycircle(3.3,3.3){-}="a"; 
(20,10)*{\otimes_\alg}; 
(20,10)*\xycircle(3.3,3.3){.}="b"; 
(20,-10)*{\balg}; 
(20,-10)*\xycircle(3.3,3.3){-}="c"; 
(14,0)*{}="d";
"a"; "d" **\dir{-}?(.6)*\dir{>};
"d"; "b" **\dir{-}?(.6)*\dir{>};
"d"; "c" **\dir{-}?(.6)*\dir{>};
(7,8)*{\Delta^\vee}; 
(35,10)*{\alg}; 
(35,10)*\xycircle(3.3,3.3){-}="aa"; 
"b"; "aa" **\dir{-}?(.6)*\dir{>}+(-1,5)*{\ell^{-1}};
(55,10)*{\alg}; 
(55,10)*\xycircle(3.3,3.3){-}="xx"; 
(70,10)*{\otimes_\alg}; 
(70,10)*\xycircle(3.3,3.3){.}="x"; 
(70,-10)*{\balg}; 
(70,-10)*\xycircle(3.3,3.3){-}="y"; 
(90,0)*{\bbc}; 
(90,0)*\xycircle(3.3,3.3){-}="z"; 
(76,0)*{}="t";
 "x";"t" **\dir{-}?(.6)*\dir{>};
 "t"; "z"**\dir{-}?(.6)*\dir{>};
 "y"; "t" **\dir{-}?(.6)*\dir{>};
 "xx"; "x" **\dir{-}?(.6)*\dir{>}+(-1,5)*{\ell};
(83,8)*{\Delta}; 
"c" ; "y" **\dir{~};
\endxy
\]
If we start tracing through the diagram at the incoming node labelled
$\alg$, which comes from the element $\ell\in \KK(\alg, \alg)$, we
obtain the expression $\ell \otimes_\alg(\Delta^\vee \otimes_{\balg}
\Delta)\otimes_\alg \ell^{-1}$. Note here that the element $\ell$ is
attached to the product $\Delta^\vee\otimes_\balg\Delta$, as we are
not allowed to change the order of the elements $\Delta^\vee$ and
$\Delta$. The calculation is now completed by using associativity
of the Kasparov product (see \cite[Proposition~2.7]{BMRS} for
details).
\end{proof}

Similarly, one gets a converse.
\begin{proposition}
Let $(\alg,\balg)$ be a PD pair, and let
$\Delta_1, \Delta_2 \in \K^d(\alg\otimes\balg)$ be fundamental 
classes with inverses $\Delta_1^\vee, \Delta_2^\vee  \in
\K_{-d}(\alg\otimes\balg)$. Then
$\Delta_1^\vee\otimes_{\balg}\Delta^{\phantom{\vee}}_2$ is an
invertible element in $\KK_0(\alg, \alg)$, with inverse given by
$(-1)^d\,\Delta_2^\vee \otimes_{\balg} \Delta^{\phantom{\vee}}_1 \in
\KK_0(\alg, \alg)$.
\end{proposition}

\begin{cor}
Let $(\alg,\balg)$ be a PD pair. Then the moduli space of
fundamental classes for $(\alg,\balg)$ is isomorphic to the group of
invertible elements in the ring $\KK_0(\alg, \alg)$.
\end{cor}

\begin{remark}
In the commutative case $\alg=C(X)$, the abelian group of units of
$\KK_0(C(X), C(X))$ is, by the universal coefficient theorem for
$\KK$-theory, an extension of the automorphism group $\Aut
(\K^\bullet(X))$ by $\Ext_\bZ (\K^\bullet(X),\K^{\bullet+1}(X))$.
\end{remark}

If $\alg$ is $\KK$-equivalent to the algebra $C_0(X)$ for some locally
compact space $X$, then $\alg^\op$ is $\KK$-equivalent to $C_0(X)^\op
= C_0(X)$ as well, and hence we have the following result.

\begin{theorem}
\label{thm:PDpairclass}
Let $\alg$ be a separable {\Ca} satisfying the universal coefficient
theorem for $\KK$-theory, and with finitely generated
$\K$-theory. Then $\alg$ is always part of a PD pair, and $\alg$ is in
addition a PD algebra if and only if either
$\rk(\K_0(\alg))=\rk(\K_1(\alg))$ \textup{(}in which case we can take
$d=1$\textup{)} or $\Tors(\K_0(\alg))\cong \Tors(\K_1(\alg))$
\textup{(}in which case we can take $d=0$\textup{)}.
\end{theorem}
\begin{proof} By hypothesis, $\alg$ is $\KK$-equi\-va\-lent to a
  commutative {\Ca}, hence we can assume $\alg$ abelian without loss
  of generality. By the universal coefficient theorem one has
\bea
\rk\big(\K^j(\alg)\big)&=&\rk\big(\K_j(\alg)\big) \ , \nonumber\\[4pt]
\Tors\big(\K_j(\alg)\big)&\cong& \Tors\big(\K^{j+1}(\alg)\big)
\nonumber
\eea
for $j=0,1$ mod~$2$. Thus the condition for $\alg$ to be a PD algebra
(\emph{i.e.}, we can take the other algebra of the PD~pair to be
$\alg^\op$) is necessary to have an isomorphism $\K_j(\alg)\to
\K^{j+d}(\alg)$. It remains to show that for $\alg$ and $\balg$
commutative, an isomorphism $\K_j(\alg)\to \K^{j+d}(\balg)$ can always
be implemented by a suitable fundamental class $\Delta$. By the
K\"unneth theorem and the universal coefficient theorem, we can build
$\Delta$ and $\Delta^\vee$ from knowledge of the $\K$-theory groups
$\K_\bullet(\alg)$, one cyclic summand at a time. Alternatively,
realize $\alg$ as $C_0(X)$ for some (possibly non-compact) manifold
$X$, take $\balg=C_0(T^*X)$, and construct $\Delta$ from the Dirac
operator on the Clifford algebra bundle of $T^*X$. When $X$ is
spin$^c$, the algebra $\balg$ is $\KK$-equivalent to $\alg=\alg^\op$.
\end{proof}

We can conclude from this last result that PD pairs are quite common.

\begin{lemma}\cite[\S7.1]{BMRS}
Let $\alg,\balg_1,\balg_2$ be separable $C^*$-alge\-bras such that
$(\alg,\balg_1)$ and $(\alg,\balg_2)$ are both PD pairs. Then
$\balg_1$ and $\balg_2$ are $\KK$-equivalent.
\label{PDKKLemma}\end{lemma}

\section{Riemann-Roch theorem for noncommutative
  spaces}\label{NCRRTheorem}

In the commutative case, the Grothendieck-Riemann-Roch theorem is a
key ingredient in arriving at the Minasian-Moore formula for D-brane
charge in Type~II superstring theory. In this section we will describe
the noncommutative version of this theorem. Later on we will connect
this formalism with the construction of D-brane charges in
noncommutative spaces.

\subsection{Grothendieck-Riemann-Roch theorem in the smooth
  case\label{GRRclassical}}

We begin by recalling the classical case, in particular the properties of
$\K$-theory that are required to define the Gysin map and to state the
Grothendieck-Riemann-Roch theorem. Let $f: N\rightarrow M$ be a smooth
map of smooth manifolds. It is said to be
\emph{$\K$-oriented} if $TN\oplus f^*(TM)$ is a spin$^c$ vector bundle over
$N$. In this case, there is a Gysin or ``wrong way'' map
\beq
  f_!\, :\,  \K^\bullet (N) ~\longrightarrow~ \K^{\bullet +d} (M) \ ,
\label{GysinK}\eeq
where $d= \dim (M) - \dim (N)$, and we regard complex
$\K$-theory as being $\bbZ_2$-graded by even and odd degree.
Under the same assumptions, there are also Gysin maps
$$
f_* \,:\, \H^\bullet_c(N, \bQ) ~\longrightarrow~ \H^{\bullet +d}_c (M, \bQ)
$$
defined in the standard way using Poincar\'e duality in cohomology
with compact supports.

Then the Grothendieck-Riemann-Roch theorem, in the version given by
Atiyah and Hirzebruch~\cite{AH}, can be phrased as the commutativity
of the diagram
$$
\begin{CD}
\K^\bullet(N) @>f_!>> \K^{\bullet+d}(M)\\
        @V{\sf Todd}(N) \smile{\ch}VV
@VV{\Todd}(M)  \smile \ch V \\
\H^\bullet_c(N, \bQ)   @>>f_*> \H^{\bullet+d}_c(M, \bQ) \ ,
\end{CD}
$$
with the bottom row $\bbZ_2$-graded by even and odd degree,
giving
\beq
\ch\big(f_!(\xi)\big) \smile {\Todd}(M) = f_*\big(
\ch(\xi)\smile {\Todd}(N)\big)
\label{GRRformclass}\eeq
for all $\xi \in \K^\bullet(N)$. Here ${\Todd}(M)\in\H_c^\bullet(M,\bQ)$
denotes the Todd characteristic class of the tangent bundle of
$M$. There are many useful variants of this beautiful formula. In the
remainder of this section we will develop a noncommutative version
which encompasses both a cohomological and a homological Riemann-Roch
theorem, along with other variants in a unified framework.  

\subsection{Todd classes\label{Toddclass}}

We will begin by recalling
from~\cite{BMRS} the definition of the Todd class of
a PD algebra $\alg$ as an element in bivariant cyclic
cohomology. Let $\underline{\mathfrak{D}}$ be the class of all
separable $C^*$-algebras $\alg$ for which there exists another
separable $C^*$-algebra $\balg$ such that $(\alg, \balg)$ is a PD
pair. For any such algebra $\alg$, we use Lemma~\ref{PDKKLemma} to fix
a representative of the $\KK$-equivalence class of $\balg$ and denote
it by $\tilde{\alg}$. In general, there is no canonical choice for
$\tilde{\alg}$. If $\alg$ is a PD algebra, then the canonical choice
$\tilde{\alg} = \alg^\op$ will always be made.

\begin{definition}
\label{defn:Todd}
Let $\alg\in\underline{\mathfrak{D}}$, let $\Delta\in\K^d(\alg\otimes
\tilde{\alg})$ be a fundamental $\K$-homology class for the pair
$(\alg, \tilde{\alg})$, and let $\Xi\in\HE^d(\alg\otimes \tilde{\alg})$
be a fundamental cyclic cohomology class. Then the {\it Todd class} of
$\alg$ is defined to be the class in the ring $\HE_{0}(\alg,\alg)$
given by
\[
\begin{split}
\Todd(\alg)~ &:=~\Xi^\vee\otimes_{ \tilde{\alg}}\ch(\Delta) \\ \\
& = \quad
\bottom{\bbc}{\alg}{\tilde{\alg}}{\Xi^\vee}{\alg}
{\tilde{\alg}}{\bbc}{\ch(\Delta)}
\label{Toddgendef}
\end{split}
\]
\end{definition}

The Todd class is invertible, with inverse given by
\[
\Todd(\alg)^{-1}=(-1)^d~\ch\big(\Delta^\vee\,\big)
\otimes_{\tilde\alg}\Xi \ .
\]
The following example provides the motivation behind this definition.

\begin{example}
Let $\alg=C(X)$ with $X$ a compact complex manifold. Then $\alg$ is a
PD algebra with fundamental class $\Delta$ given by the Dolbeault
operator on $X\times X$. After passage from $\alg$ to the dense
subalgebra $C^\infty(X)$, we can identify $\HE$ with the usual
periodic cyclic homology $\HP$. Thus $\HE_0(\alg,\alg)$ can be
identified with $\End(\H^\bullet(X,\bQ))$ (see
eq.~(\ref{HLHomiso})). The natural choices of fundamental classes
$\Xi$ and $\Xi^\vee$ come from the usual Poincar\'e duality in
rational cohomology using the orientation cycle $[X]$. Then
$\Todd(\alg)$ is just cup product with the usual Todd cohomology class
$\Todd(X)\in\H^\bullet(X,\bQ)$.
\end{example}

\subsection{$\K$-oriented morphisms\label{Koriented}}

One of the most interesting applications of the present formalism is
a definition of Gysin or ``wrong way'' homomorphisms, following Connes
and Skandalis~\cite{CS}. More
generally, we can now study {$\K$-oriented morphisms}. If
$f\colon\alg\to\balg$ is a morphism of {\Ca}s in a suitable category,
a $\K$-orientation is a functorial way of assigning a corresponding
element $f! \in \KK_d(\balg, \alg)$. (Note that this element
points in the \emph{opposite} direction from $f$.)  The
Gysin homomorphism is then given by
$$
f_!:=(-)\otimes_\balg f!\,:\,\K_\bullet(\balg)~\longrightarrow~
\K_{\bullet+d}(\alg) \ .
$$

If $\alg$ and $\balg$ are both PD algebras, then any morphism
$f\colon\alg\to\balg$ is $\K$-oriented and the element $f!$ is
determined as 
\[
\begin{split}
f! ~ & =~ (-1)^{d_\alg}\,
\Delta_\alg^\vee \otimes_{\alg^\op} [f^\op]_{\KK} \otimes_{\balg^\op}
  \Delta_\balg \\ \\ & = (-1)^{d_\alg}\quad
\xy
(0,0)*{\bbc}; 
(0,0)*\xycircle(3.3,3.3){-}="a"; 
(20,10)*{\alg}; 
(20,10)*\xycircle(3.3,3.3){-}="b"; 
(20,-10)*{\alg^\op}; 
(20,-10)*\xycircle(3.3,3.3){-}="c"; 
(14,0)*{}="d";
"a"; "d" **\dir{-}?(.6)*\dir{>};
"d"; "b" **\dir{-}?(.6)*\dir{>};
"d"; "c" **\dir{-}?(.6)*\dir{>};
(7,8)*{\Delta^\vee_\alg}; 
(35,10)*{\balg}; 
(35,10)*\xycircle(3.3,3.3){-}="x"; 
(35,-10)*{\balg^\op}; 
(35,-10)*\xycircle(3.3,3.3){-}="y"; 
(55,0)*{\bbc}; 
(55,0)*\xycircle(3.3,3.3){-}="z"; 
(41,0)*{}="t";
 "x";"t" **\dir{-}?(.6)*\dir{>};
 "t"; "z"**\dir{-}?(.6)*\dir{>};
 "y"; "t" **\dir{-}?(.6)*\dir{>};
(48,8)*{\Delta_\balg}; 
(22.7,-10)*{}="cc";
(32.5,-10)*{}="yy"; 
"cc" ; "yy" **\dir{-}?(.6)*\dir{>}+(0,4)*{f^\op};
\endxy
\end{split}
\]
with $d=d_\alg-d_\balg$. To check functoriality of this construction,
observe that if $\alg$, $\balg$ and $\calg$ are PD algebras, and if
$f:\alg\to\balg$, $g:\balg\to\calg$ are morphisms of $C^*$-algebras,
then
\bea
&& \bigl( (-1)^{d_\alg}\, \Delta^\vee_\alg \otimes_{\alg^\op}
[f^\op]_{\KK}\otimes_{\balg^\op} \Delta_\balg \bigr)
~\otimes_\balg~\bigl( (-1)^{d_\balg}\,
\Delta^\vee_\balg \otimes_{\balg^\op} [g^\op]_{\KK} \otimes_{\calg^\op}
\Delta_\calg \bigr) \nonumber \\ && \qquad\qquad \qquad\qquad 
=~ (-1)^{d_\alg}\, \Delta^\vee_\alg \otimes_{\alg^\op} \big[(g\circ
f)^\op\big]_{\KK} \otimes_{\calg^\op} \Delta_\calg \nonumber
\eea
by associativity of the Kasparov product and the basic
relation 
\[
(-1)^{d_\balg}\,\Delta^\vee \otimes_\balg\Delta = 1_{\balg^\op} \ .
\]  
This can again be checked using the diagram calculus.

The following examples illustrate the connection with the classical
notion of $\K$-orientation given in~\S\ref{GRRclassical} above.

\begin{example}

Let $h: X \hookrightarrow Y$ be a $\K$-oriented smooth
embedding of smooth compact manifolds. Since $TX
\oplus TX$ has a canonical almost complex structure,
the normal bundle $$N_Y X = h^*(TY)/TX$$ is a spin$^c$ vector
bundle. In particular, the zero section embedding $\iota^X : X
\hookrightarrow N_Y X$ is $\K$-oriented. Let $S(N_YX)$ denote the
bundle of spinors associated to the spin$^c$ structure on $N_YX$. Let
$\mathcal{E}=(\mathcal{E}_x)_{x\in X}$, $\mathcal{E}_x=L^2(N_x,S_x)$
be the Hilbert bundle over $N=N_YX$ obtained from
the pullback $S=\pi_N{}^*S(N_YX)$, where $\pi_N:N_YX\to X$ is the bundle
projection. For $\xi\in C(X)$, the composition $\xi\circ\pi_N$ acts by
multiplication as an endomorphism of $\mathcal{E}$ and thus there is a
natural homomorphism $C(X)\to{\rm End}(\mathcal{E})$. When the
codimension $n = \dim(Y) -\dim(X)$ of $X$ in $Y$ is even, Clifford
multiplication by the orientation of $N_YX$ splits $S(N_YX)$ into a
Whitney sum of half-spin bundles $S(N_YX)^\pm$ which define a
$\bbZ_2$-grading $\mathcal{E}=\mathcal{E}^+\oplus\mathcal{E}^-$. Clifford
multiplication $c(v) :\pi_N{}^*S(N_YX)^+ \to \pi_N{}^*S(N_YX)^-$ by
the tautological section $v$ of the bundle
$$\pi_N{}^*N_YX~\longrightarrow~N_YX \ , $$ which assigns to a vector
in $N_YX$ the same vector in $\pi_N{}^*N_YX$, defines a morphism
$$F_{v(\nu)}~\in~\Hom\big(\mathcal{E}_{v(\nu)}^+\,,\,\mathcal{E}_{v(\nu)}^-\big)$$
for all $\nu\in N_YX$.
The corresponding Kasparov bimodule $(\mathcal{E},F)$ defines an
invertible element ${\iota^X}!\in\KK_n(C(X),C_0(N_YX))$ associated to
the classical Atiyah-Bott-Shapiro (ABS) representative of the Thom class of
the complex vector bundle $N_YX$. Upon choosing a
riemannian metric on $Y$, there is a diffeomorphism $\Phi$ from a
tubular neighbourhood $U$ of $h(X)$ onto a neighbourhood of the zero
section in the normal bundle $\iota^X(X)$, giving an invertible element
$[\Phi]_{\KK}\in\KK_0(C_0(N_Y X),C_0(U))$. For any open subset $j : U
\hookrightarrow Y$, the extension by zero defines an element
$j!\in\KK_0(C_0(U),C(Y))$. Then we may associate to the embedding $h$
the class in $\KK$-theory defined by
$$
h! = {\iota^X}!\otimes_{C_0(N_YX)}[\Phi]_{\KK}\otimes_{C_0(U)}j!~\in~
\KK_n\big(C(X)\,,\,C(Y)\big)
$$
which, by homotopy invariance of $\KK$-theory and functoriality of
Gysin classes, is independent of the choices made. In the applications
to string theory, this construction establishes that the charge of a
D-brane supported on $X$ takes values in the $\K$-theory of
spacetime~$Y$~\cite{OS99,W}.

\label{Gysinemb}\end{example}

\begin{example}

Let $\pi : Y \to Z$ be a $\K$-oriented, proper smooth submersion of
smooth compact manifolds. Since every smooth compact manifold $Y$ has
a smooth embedding $h:Y\hookrightarrow\RR^{2q}$ for $q$
sufficiently large, a parametrized version yields a smooth embedding
$\kappa=(\pi,h): Y \hookrightarrow Z \times \RR^{2q}$ which is
$\K$-oriented. The corresponding $\KK$-theory class is
$\kappa!\in\KK_a(C(Y),C_0(Z \times\RR^{2q}))$, where $a = \dim(Z)
-\dim(Y)+2q $. Let $\iota^Z: Z\hookrightarrow Z \times \RR^{2q}$
denote the zero section embedding, with invertible Thom class
${\iota^Z}!\in\KK_{2q}(C(Z),C_0(Z\times\RR^{2q}))$. Then the
$\KK$-theory class corresponding to the submersion $\pi$ is defined as
$$
\pi! = \kappa! \otimes_{C_0(Z\times
\RR^{2q})}\big({\iota^Z}!\big)^{-1}~\in~
\KK_b\big(C(Y)\,,\,C(Z)\big) 
$$
where $b = \dim(Z) - \dim(Y)$, and is again independent of the choices
made. When $Y=X\times Z$ is a product manifold, the corresponding
Gysin map $\pi_!$ gives the analytic index for $Z$-families of
elliptic operators on $X$.

\end{example}

\begin{example}

Let $f : N \to M$ be an arbitrary $\K$-oriented smooth
proper map, as in the setting of~\S\ref{GRRclassical} above. The
map $f$ can be canonically factored as $f=p^M\circ{\rm gr}(f)$,
firstly into the $\K$-oriented smooth graph embedding ${\rm gr}(f) : N
\hookrightarrow N\times M$ defined by ${\rm gr}(f)(x) = (x,
f(x))$. The construction of Example~\ref{Gysinemb} above then gives an
element ${\rm gr}(f)!\in\KK_m(C(N),C(N \times M))$ with $m =
\dim(M)$. Secondly, the projection $p^M : N\times M \to M$ is a
$\K$-oriented proper submersion, when restricted to the image of ${\rm
  gr}(f)$. The corresponding $\KK$-theory class is $${p^M}!~\in~\KK_b\big(C(M
\times N)\,,\,C(M)\big)$$ where $b = - \dim(N)$. The Gysin map of $f$ in
eq.~(\ref{GysinK}) is then defined via the $\KK$-theory element
$$
f! = {\rm gr}(f)!\otimes_{C(N\times M)}{p^M}! \ .
$$

\end{example}

\begin{example}
\label{ex:quasihomo}
Let $\jalg$, $\alg$ and $\balg$ be separable {\Ca}s, and suppose there
is a split short exact sequence
\begin{equation}
\label{eq:sses}
\xymatrix{
0~\ar[r] & ~\jalg~ \ar[r]^j & ~\alg~ \ar[r]_q & ~\balg ~
\ar@/_/[l]_s \ar[r] & ~0 \ . }
\end{equation}
Then the morphism $j$ is naturally $\K$-oriented.  Indeed,
by the split exactness property of $\KK$ from~\S\ref{sec:KKaxioms},
$j$ and $s$ induce an isomorphism
\[
\Psi=(j_* \oplus s_*)\,\colon\,  \KK(\alg,\jalg) \oplus \KK(\alg,\balg)
~\xrightarrow{\approx}~\KK(\alg,\alg) \ .
\]
We define $j!\in \KK(\alg,\jalg)$ to be the image of $1_\alg$ in
$\KK(\alg,\jalg)$ under $\Psi^{-1}$ followed by projection
onto the first summand in the direct sum, and observe that it has the
desired properties of a $\K$-orientation.  This is the same as the
Kasparov element called $\pi_s$ in \cite[\S17.8]{Black}.
This choice of $\K$-orientation depends on the splitting $s$,
which is the reason for the notation.
\end{example}

\subsection{The noncommutative Grothendieck-Riemann-Roch
  theorem\label{NCGRR}}

The Grothendieck-Rie\-mann-Roch formula compares 
the two bivariant cyclic cohomology classes ${\ch}(f!)$ and $f*$.
\begin{theorem}
\label{thm:GRR}
Suppose that $\alg$ and $\balg$ are PD algebras with given $\HE$
fundamental classes. Let $f*\in \HE(\balg,\alg)$ be the Gysin map
defined using the $\HE$ fundamental classes.
Then one has the Grothendieck-Rie\-mann-Roch formula 
\begin{equation}
{\ch}(f!) = (-1)^{d_\balg}~
\Todd(\balg) \otimes_{\balg} (f* \,) \otimes_{\alg}\,
\Todd(\alg)^{-1} \ . 
\label{eqn:GRR-strong}
\end{equation}
\end{theorem}

\begin{proof}
We write out both sides of eq.~\eqref{eqn:GRR-strong} using the
various definitions, and simplify using associativity of the
Kasparov product and the functorial properties of the bivariant
Chern character (see~Theorem~\ref{ChernHLthm}). After some algebra,
one finds that the theorem then follows if 
\bea
\bigl( \Xi_\balg^\vee\otimes_{\balg^\op}\ch(\Delta_\balg)\bigr)
\otimes_{\balg} \Xi_\balg 
&=& (-1)^{d_\balg}~ \ch(\Delta_\balg) \ , \nonumber \\[4pt]
\Xi_\alg^\vee \otimes_{\alg}  \bigl(\ch(\Delta_\alg^\vee\,)
\otimes_{\alg^\op}\Xi_\alg\bigr)
&=& (-1)^{d_\alg}~\ch(\Delta^\vee_\alg) \ . \nonumber
\eea
Both of these equalities follow easily from the diagram
calculus. See~\cite[Theorem~7.10]{BMRS} for further details.
\end{proof}

\begin{cor}
 \ ${\sf ch}\big(f_!(\xi)\big) \otimes_{\alg}  {\sf Todd}(\alg)= 
(-1)^{d_\balg}\,f_* \bigl({\sf ch}(\xi) \otimes_{\balg} 
{\sf Todd}(\balg) \bigr)$  \ for all $\xi\in \K_\bullet(\balg)$.
\label{Gysincor}\end{cor}

\begin{example}

Let $\pi:X\to Z$ be a smooth fibration over a closed smooth manifold
$Z$ whose fibres $X/Z$ are compact, closed spin$^c$ manifolds of even
dimension. Let $g^{X/Z}$ be a metric on the vertical tangent
sub-bundle $T(X/Z)$ of $TX$. Let $S_{X/Z}$ be the spinor bundle
associated to $(T(X/Z),g^{X/Z})$. Let $T^{\rm H}X$ be a horizontal
vector sub-bundle of $TX$ such that $TX=T^{\rm H}X\oplus T(X/Z)$. Then
$(T^{\rm H}X,g^{X/Z})$ determines a canonical euclidean connection
$\nabla^{X/Z}$~\cite[Theorem~1.9]{bismut}. Clifford multiplication of
$T^*(X/Z)$ on $S_{X/Z}$ and the connection $\nabla^{X/Z}$ define a
fibrewise Dirac operator $\Dirac_z$ acting on the Hilbert space
$\bun_z=L^2(X/Z,S_{X/Z})$ along the fibre $\pi^{-1}(z)\cong X/Z$ for
$z\in Z$. Then $\{\Dirac_z\}_{z\in Z}$ is a smooth family of
elliptic operators parametrized by $Z$ acting on an
infinite-dimensional, $\bZ_2$-graded Hilbert bundle $\bun\to Z$ whose
fibre at $z\in Z$ is $\bun_z$. The corresponding Kasparov bimodule
defines the longitudinal Dirac element
$\pi!\in\KK(C_0(X),C_0(Z))$~\cite{CS}. Given a complex vector bundle
$E\to X$, the Atiyah-Singer index theorem asserts that
$\pi_!(E)=E\otimes_{C_0(X)}\pi!$ is equal to the analytic index of the
family of Dirac operators on $X/Z$ coupled to
$E$. Corollary~\ref{Gysincor} then expresses the Chern character of
the index in $\K$-theory in terms of the Chern character of $E$ as
$$
\ch\big(\pi_!(E)\big)=\pi_*\big(\Todd(X/Z)\smile\ch(E)\big) \ .
$$
This quantity is used to compute global worldsheet anomalies in string
theory~\cite{FW}.
\end{example}

\section{D-brane charge in noncommutative spaces}\label{sect:dbrane}

In this section we will apply our formalism to the description of
D-brane charges in very general noncommutative settings. After
recalling the classical situation, we will describe the physical
context in which the $\KK$-theory of $C^*$-algebras classifies states
of D-branes in superstring theory. We will then derive the
noncommutative formula for Ramond-Ramond charges and present some
noncommutative examples.

\subsection{The Minasian-Moore formula}

We begin by recalling the classical formula. Let $X$ be a compact
spin$^c$ manifold of dimension $d$. Then Poincar\'{e} duality in
ordinary cohomology of $X$ is the statement that the pairing given by
the cup product
\beq
(x,y)_\H = \big\langle x\smile y\,,\, [X]\big\rangle
\label{cohpairingclass}\eeq
for $x\in\H^\bullet(X,\zed)$, $y\in\H^{d-\bullet}(X,\zed)$ is
non-degenerate. On the other hand, in $\K$-theory the analytic index
provides another pairing given by
\[ 
(E,F)_\K = \text{index}(\Dirac_{E\otimes F})
\]
where $E,F$ are classes in $\K^0(X)$ represented by vector
bundles over $X$, and $\Dirac$ is the Dirac operator
associated with the $\spin^c$ structure on $X$.

The classical Chern character gives a ring isomorphism
$$
\ch \,:\, \K^j(X) \otimes \bbq ~\longrightarrow~ \H^j(X, \bbq)
$$
where $j=0, 1$ and on the right-hand side we mean the periodised
cohomology of $X$ with rational  coefficients. This isomorphism is not
compatible with the two pairings. However, the Atiyah-Singer index
theorem provides the formula
\[
\text{index}(\Dirac_{E\otimes F}) = \big\langle \Todd (X) \smile
\ch(E\otimes F)\,,\, [X] \big\rangle
\]
from which it follows that the modified Chern character $\Ch$ defined
by the generalized Mukai vector $\Ch(E) = \sqrt{\Todd(X)} \smile
\ch(E)$ is an isometry with respect to the index pairing on
$\K$-theory and the topological pairing on cohomology. This simple
observation is the starting point of~\cite{BMRS} and plays an
important role in this section.

\begin{definition}
  An ({\it untwisted}) {\it D-brane} in $X$ is a triple $(W,E,f)$,
  where $f:W\hookrightarrow X$ is a closed, embedded spin$^c$
  submanifold and $E\in\K^0(W)$. The submanifold $W$ is called the
  {\it worldvolume} and the class $E$ the {\it Chan-Paton bundle} of the
  D-brane.
\label{Dbranedef}\end{definition}

From Example~\ref{Gysinemb} it follows that any D-brane $(W,E,f)$ in
$X$ determines a canonically defined $\KK$-theory class
$f!\in\KK(C(W),C(X))$. In fact, every D-brane naturally determines a
Fredholm module over the algebra $C(X)$ and we can think of D-branes
$(W,E,f)$ as providing $\K$-homology classes on spacetime $X$,
Poincar\'e dual to $\K$-theory classes~$f_!(E)$. In this context,
Baum-Douglas equivalence~\cite{BD} (or stable homotopy equivalence) is
sometimes refered to as ``gauge equivalence'' between D-branes.

\begin{definition}
The {\it Ramond-Ramond charge} of a D-brane
  $(W,E,f)$ in $X$ is the modified Chern characteristic class 
  $${\sf Q}(W,E)~:=~{\sf
  Ch}\big(f_!(E)\big)~ =~ \ch\big(f_!(E)\big) \smile 
\sqrt{\Todd(X)}\in\H^\bullet(X,\rat) \ .
  $$
\end{definition}

This is the Minasian-Moore formula~\cite{MM}. In the underlying
boundary superconformal field theory, the charge vector ${\sf Q}(W,E)$
is the zero-mode part of the boundary state of the D-brane in the
Ramond-Ramond sector. In the setting of the D-brane field theory on
the worldvolume $W$, ${\sf Q}(W,E)=f_*({\sf D}_{\rm WZ}(W,E))$ is the
image of the \emph{Wess-Zumino class} ${\sf D}_{\rm
  WZ}(W,E)\in\H^\bullet(W,\rat)$ (for vanishing $B$-field) under the
Gysin map in cohomology. From the Grothendieck-Riemann-Roch formula
(\ref{GRRformclass}) and naturality of the Todd class, one has
\beq
{\sf D}_{\rm WZ}(W,E)=\ch(E)\smile\sqrt{\Todd(W)/\,\Todd(N_XW)} \ .
\label{WZclassical}\eeq
This formula characterizes the Ramond-Ramond charge as the anomaly
inflow on the D-brane worldvolume, determined by the index theorem on
$W$.

\subsection{Algebraic characterization of D-branes}

We will now describe a physical framework for $\KK$-theory which will
naturally apply to the construction of D-brane charges in more general
situations. (See~\cite{Asakawa:2001vm,Periwal:2000eb} for related but
somewhat differently motivated characterizations.) Let us fix a closed
string background specified by a compact spin$^c$ manifold $X$. In the
absence of background supergravity form fields, this is
equivalent to fixing the commutative $C^*$-algebra $\alg=C(X)$. An
open string in $X$ may be regarded as an embedding of a compact
Riemann surface $\Sigma\hookrightarrow X$, called the
\emph{worldsheet}, with boundary $\partial\Sigma$. A classical D-brane
then corresponds to a choice of submanifold $W\subset X$ such that the
open string fields define relative maps
$$(\Sigma,\partial\Sigma)~\longrightarrow~(X,W) \ . $$

Alternatively, we may regard an open string as an oriented embedding
of the interval $I=[0,1]$ into $X$, with the boundaries $t=0$ and
$t=1$ refering to the {\it endpoints} of the strings. The worldsheet
is then $\Sigma=\real\times I$. In the associated boundary
superconformal field theory, we require that the Cauchy problem for
the Euler-Lagrange equations on $\Sigma$ have a unique solution
locally. This requires imposing suitable boundary conditions on the
open string fields. A classical D-brane simply represents such a
choice of boundary conditions. To fully specify the boundary
conditions one must also specify a hermitean vector bundle $E$ over
$W$ with connection. The requirement that the boundary conditions
preserve superconformal invariance constrains the submanifold $W$. For
example, in the absence of $H$-flux the worldvolume $W$ must be
spin$^c$~\cite{FW} and we recover Definition~\ref{Dbranedef} above of
a D-brane.

We would now like to extend this description to {\it quantum}
D-branes, those which define consistent boundary conditions after
quantization of the boundary superconformal field
theory. Unfortunately, there is no known satisfactory way to define a
fully quantum boundary condition in these generic instances. In the
following we will describe a conjectural $C^*$-algebraic framework for
quantum D-branes in the context of open string field theory.

Let $\mathfrak{B}$ be the set of boundary conditions for open strings
in $X$. For $a,b\in\B$, we call an open string with $a$ boundary
conditions at its $t=0$ end and $b$ boundary conditions at its $t=1$
end an \emph{$a$-$b$ open string}. One can glue two open strings
together by the usual concatenation of paths $I\hookrightarrow X$,
provided that the boundary condition at the initial endpoint of one
string matches the boundary condition at the final endpoint of the
other string. After quantization, this can be used to define a
noncommutative algebra of observables $\dalg_a$ consisting of $a$-$a$
open string fields~\cite{SW,WOverview}, thus generating a set of
D-brane algebras $\dalg_a$, $a\in\B$. These algebras are required to
carry an action of the worldsheet superconformal algebra by
automorphisms of $\dalg_a$. By concatenation, for any pair of boundary
conditions $a,b\in\B$ the $a$-$b$ open string fields form a
$\dalg_a$-$\dalg_b$ bimodule $\bun_{a,b}$. Similarly, the $b$-$a$ open
string fields generate a $\dalg_b$-$\dalg_a$ bimodule $\bun_{b,a}$. In
particular, $\bun_{a,a}=\dalg_a$ is the trivial $\dalg_a$-bimodule
given by the natural actions of the algebra $\dalg_a$ on itself via
multiplication from the left and from the right.

We would now like to associate to this data a $\complex$-linear
category that we may wish to call the ``category of D-branes''. The
objects of this category are the elements of $\mathfrak{B}$, and the
set of morphisms from $a$ to $b$ is the bimodule $\bun_{a,b}$ (or more
precisely the corresponding state space of the boundary superconformal
field theory). Given some other boundary condition $c\in\B$, we must
then require that there is an associative bilinear composition map
\beq
\bun_{a,b}\times\bun_{b,c}~\longrightarrow~\bun_{a,c} \ .
\label{vertexmap}\eeq
The naive guess for this map is the natural string vertex combining an
$a$-$b$ open string field and a $b$-$c$ open string field to generate
an $a$-$c$ open string field. However, the problem is that this map
need not be well-defined. Recall that elements of the open string
bimodule $\bun_{a,b}$ are the vertex operators
$V_{a,b}:I\to\End(\hil_{a,b})$ with $\hil_{a,b}$ a separable Hilbert
space. The product of $V_{a,b}$ and $V_{b,c}$ is encoded in a singular
operator product expansion which for $t>t'$ takes the formal symbolic
form
\beq
V_{a,b}(t)\cdot V_{b,c}(t'\,)=\sum_{j=1}^N\,\frac1{(t-t'\,)^{h_j}}~
W_{a,b,c|j}(t,t'\,)
\label{OPE}\eeq
for some positive integer $N$, where
$W_{a,b,c|j}:I\times I\to\End(\hil_{a,c})$, the real numbers
$h_j\in[0,\infty)$ are known as conformal dimensions, and the quantity
$(t-t'\,)^{-h_j}$ is understood in terms of its formal power series
expansion
$$
\frac1{(t-t'\,)^{h_j}}=\sum_{n=0}^\infty\,\frac{h_j\,(h_j+1)\cdots(h_j+n-1)}
{n!}~(t'\,)^n\,t^{-n-h_j}
$$
for $t>t'$. When the conformal dimensions are non-zero, the leading
singularities of the operator product expansion do not give an
associative algebra in the standard sense.

One resolution to this problem was pointed out by Seiberg and
Witten~\cite{SW}. They consider a particular limit of the
boundary conformal field theory on $X$ whereby the dimensions $h_j$
vanish. Quantization of the point particle at the endpoint of an open
string with boundary condition $a\in\B$ gives a separable Hilbert
space $\hil_a$ which is acted upon by open string tachyon vertex
operators $V_a(t)$, $t\in[0,1]$. In the Seiberg-Witten limit, these
operators live in a separable $C^*$-algebra $\dalg_a$ representing a
noncommutative spacetime. (This algebra is a deformation of
$\alg=C(X)$ when there is a constant $B$-field present on $X$.) The
product $\dalg_a\otimes\dalg_b$ generates the full algebra of operators
on the string ground states, acting irreducibly on the quantum
mechanical Hilbert space $\hil_{a,b}=\hil_a\otimes\hil_b^\vee$. In
this case, the mapping (\ref{vertexmap}) is well-defined and can be
computed from eq.~(\ref{OPE}), which takes an $a$-$b$ vertex operator
$V_{a,b}(t)$ and a $b$-$c$ vertex operator $V_{b,c}(t'\,)$ to the
$a$-$c$ vertex operator $V_{a,c}(t'\,)$ given by the product
$$
V_{a,c}(t'\,)=\lim_{t\to t'}\,V_{a,b}(t)\cdot V_{b,c}(t'\,) \ .
$$
Furthermore, by associativity of the operator product expansion in the
limit, $$(V_{a,b}\,V_{b})\,V_{b,c}=V_{a,b}\,(V_{b}\,V_{b,c})$$ for
any $b$-$b$ vertex operator $V_{b}$, the assignment extends to a map
\beq
\bun_{a,b}\otimes_{\dalg_b}\bun_{b,c}~\longrightarrow~\bun_{a,c} \ .
\label{vertexmapassoc}\eeq

It follows that the bimodule $\bun_{a,b}$ in this instance establishes
a (strong) Morita equivalence between the $C^*$-algebras $\dalg_a$ and
$\dalg_b$, with dual $\bun_{a,b}^\vee\cong\bun_{b,a}$. (The opposite
algebra $\dalg_{a}^\op$ is generated by reversing the orientations of
the $a$-$a$ open strings.) Furthermore, we can define elements
$$\alpha_{a,b}~\in~\KK(\dalg_a,\dalg_b) \qquad\mbox{and}\qquad
\alpha_{b,a}~\in~\KK(\dalg_b,\dalg_a)$$ by the equivalence classes of
the Kasparov bimodules $(\bun_{a,b},0)$ and $(\bun_{b,a},0)$,
respectively. Since one can canonically identify the
$\dalg_a$-bimodule $$\bun_{a,b}\otimes_{\dalg_b}\bun_{b,a}$$ with
$\dalg_a$ and the $\dalg_b$-bimodule
$$\bun_{b,a}\otimes_{\dalg_a}\bun_{a,b}$$ with $\dalg_b$, one has
$$\alpha_{a,b}\otimes_{\dalg_b}\alpha_{b,a}~=~1_{\dalg_a}
\qquad\mbox{and}\qquad
\alpha_{b,a}\otimes_{\dalg_a}\alpha_{a,b}~=~1_{\dalg_b}\ . $$ It follows
that the D-brane algebras $\dalg_a$ and $\dalg_b$ in this instance are
$\KK$-equivalent. From the above arguments it also follows that the
$\KK$-equivalence
$\alpha_{a,c}=\alpha_{a,b}\otimes_{\dalg_b}\alpha_{b,c}$ is associated
to the Morita equivalence between $\dalg_a$ and $\dalg_c$. This
involutive sort of $\KK$-equivalence is the categorical statement of
T-duality of the noncommutative spacetimes represented by these
algebras~\cite{BMRS,SW}.

In the general case, the operator product expansion (\ref{OPE}) is not
associative, and it need not even lead to a well-defined map
(\ref{vertexmap}). In particular, $\bun_{a,b}$ need not be a Morita
equivalence bimodule. There are special instances when the desired
categorical construction goes through without the need of taking any
limits. When the worldsheet superconformal field theory on
$\Sigma=\real\times I$ is a two-dimensional topological field theory,
then the set of D-branes has the structure of a category called the
category of ``topological
D-branes''~\cite{Aspinwall:2004jr,Moore:2006dw}. To extend the
construction to more generic situations, we will assume that generally
there is an appropriate ``completion'' of $\bun_{a,b}$ to a Kasparov bimodule
$(\bun_{a,b},F_{a,b})$, defining an element in $\KK(\dalg_a,\dalg_b)$,
with the mapping (\ref{vertexmapassoc}) given by the Kasparov
composition product. The canonical duality above is generically
lost, as the $\KK$-theory classes of $(\bun_{a,b},F_{a,b})$ and
$(\bun_{b,a},F_{b,a})$ need not even be inverses of each other. When they
are, the D-brane algebras $\dalg_a$, $\dalg_b$ are
$\KK$-equivalent. If the additional application of this
$\KK$-equivalence is of order two up to Morita equivalence, then the
D-brane algebras are said to be \emph{T-dual}, in the spirit
of~\cite{BMRS}. In~\S\ref{CorrTduality} we will encounter a number of
examples of T-dual algebras in this sense. The correspondence picture
there will relate these dualities to generalizations of the (smooth)
Fourier-Mukai transform and show that the proper physical
interpretation of $\KK$-theory groups is that of generalized morphisms
between open string algebras (boundary conditions) which provide
dualities among the corresponding noncommutative spacetimes. It is not
clear whether or not this fits into Witten's noncommutative
$*$-algebras~\cite{WittenSFT}, as in our formulation the manifest
background independence of open string field theory, represented by
stable isomorphisms
$\dalg_a\otimes\mathcal{K}\cong\dalg_b\otimes\mathcal{K}$ for
$a,b\in\B$, is generically lost.

In this formalism, it is more natural to think of the D-branes
themselves as algebras $\dalg_a$, rather than as Fredholm modules over
the spacetime algebra $\alg$. String theory then requires some extra
structure. We assume that for each $a\in\B$ there are canonical
elements $$\lambda_a~\in~\KK(\dalg_a,\alg) \qquad\mbox{and}\qquad
\xi_a~\in~\KK(\alg,\dalg_a) \ , $$ which determine maps between the
$\K$-theory and $\K$-homology of $\alg$ and of $\dalg_a$ via the
composition product as explained in~\S\ref{CompProd}. Passing
this structure over to $\HE$-theory, this is simply the construction
of a boundary state for the D-brane in local cyclic homology
$\HE_\bullet(\alg)$ of the closed string background, now including all
excited oscillator states of
the strings and not just their zero-modes. The boundary state is not
arbitrary but satisfies conditions of modular invariance in the
boundary conformal field theory, dictated through the Cardy
conditions. Presumably this is the essence of the generalized
structure of the morphisms between D-brane algebras in bivariant
theories, although in this crude target space picture we have not been
able to make this more precise in the general case. One instance where
the Cardy condition can be written down is the case when a pair of
D-brane algebras $\dalg_a,\dalg_b$ are $\KK$-equivalent, implemented
by a class $$\alpha_{ab}~\in~ \KK(\dalg_a,\dalg_b)$$ with inverse
$$\alpha_{ab}^{-1}~\in~\KK(\dalg_b,\dalg_a) \ . $$ Then the Cardy
condition in the present context can be stated as the requirement that
the map $\KK(\dalg_a,\dalg_a)\to\KK(\dalg_b,\dalg_b)$ induced by the
elements $\lambda_b\otimes_\alg\xi_a$ and
$\lambda_a\otimes_\alg\xi_b$ is given by the ``similarity
transformation''
$$
\big(\lambda_b\otimes_\alg\xi_a\big)\otimes_{\dalg_a}\psi
\otimes_{\dalg_a}\big(\lambda_a\otimes_\alg\xi_b\big)=
\alpha_{ab}^{-1}\otimes_{\dalg_a}\psi\otimes_{\dalg_a}\alpha_{ab}
$$
for all $\psi\in\KK(\dalg_a,\dalg_a)$.

Below we will assume that this picture extends to more general
situations in which spacetime is described by a noncommutative
separable $C^*$-algebra $\alg$. Our original definition of D-brane can
be easily cast into this framework.

\begin{example}

Consider a D-brane $(W,E,f)$ in $X$ in the sense of
Definition~\ref{Dbranedef}, and let $$\dalg~=~C(W)
\qquad\mbox{and}\qquad \alg~=~C(X) \ . $$ In the Ramond sector, the
ground states of $W$-$W$ open strings, with both ends on $W$, generate
the $\dalg$-bimodule $\bun_{W,W}=L^2(W,S(W)\otimes S(N_XW))$, where
$S(W)$ is the bundle of spinors on $W$ and $S(N_XW)$ is the bundle of
spinors on the normal bundle $N_XW$ to $W$ in $X$. Using the Dirac
operator associated to $S(W)\otimes S(N_XW)$ and pushing forward with
the group homomorphism induced by $f$, we obtain a $\K$-homology class
in $\K^\bullet(\alg)=\K_\bullet(X)$. However, we can also consider a
``spacetime filling'' D-brane $(X,\id_X^\complex,\Id_X)$, with
$\id_X^\complex$ the trivial complex line bundle over $X$, and open
strings with one end on $W$ and the other end on $X$. Then the Hilbert
space of ground state Ramond sector $X$-$W$ strings is the
$\alg$-$\dalg$ bimodule $\bun_{X,W}=L^2(W,S(W)\otimes E)$, giving a
class $\lambda_W\in\KK(\dalg,\alg)$ in the manner explained
above. When $E\cong\id_W^\complex$, this is the element $f!$
constructed in Example~\ref{Gysinemb}, while the element
$\xi_W\in\KK(\alg,\dalg)$ is the class $[f^*]_{\KK}$ of the
induced pullback morphism $$f^*\,:\,\alg~\longrightarrow~\dalg \ . $$
The equivalence between these two descriptions of the D-brane is
essentially the process of tachyon condensation on spacetime filling
branes to stable states of lower-dimensional D-branes, as explained
in~\cite{OS99,ReisSz,W}, represented by the $\KK$-theory class
$f!$. Generic classes in $\KK(\dalg,\alg)$ can be interpreted in terms
of ``generalized brane decays'' as described in~\cite{RSV}.

\label{spinorex}\end{example}

\subsection{The isometric pairing formula} 

Just as in the commutative case, Poincar\'{e} duality leads to a
nondegenerate pairing in $\K$-theory (modulo torsion). Since
Poincar\'{e} duality was formulated in terms of classes in bivariant
$\KK$-theory, there is a corresponding pairing in $\K$-homology as
well. Moreover, once we have the correct notion of Poincar\'{e}
duality in bivariant cyclic homology, we are able to define a pairing
on ordinary cyclic homology and cohomology. Below we assume that the
algebra $\alg$ is a PD algebra, but our results all generalize
straightforwardly to more general algebras in the class
$\underline{\mathfrak{D}}$ introduced in~\S\ref{Toddclass}, with
the opposite algebras $\alg^\op$ replaced everywhere by $\tilde\alg$.

Let $\alpha\in\K_i(\alg)$, $\beta\in\K_{-d-i}(\alg)$. Then there
is a pairing $(-, -)_\K : \K_i(\alg) \times
\K_{-d-i}(\alg) \rightarrow \bbz$ given by
\beq\label{pairing}
\begin{aligned}
(\alpha,\beta)_\K~&=~(\alpha\otimes\beta^\op)
\otimes_{\alg\otimes \alg^\op}\Delta\\ & \\
& = \quad \both{\bbc}{\alg}{\alg^\op}{\alpha\otimes \beta^\op}{\alg}
{\alg^\op}{\bbc}{\Delta}
~\in~\KK_0(\bbC, \bbC)=\zed \ .
\end{aligned}
\eeq
If the fundamental class $\Delta$ is symmetric then this defines a
symmetric form. This formula generalizes the index pairing in
$\K$-theory. If we take $\alg=\alg^\op =C(X)$ for $X$ a compact
spin$^c$ manifold and let $\Delta=\Dirac\otimes\Dirac$ be
the Dirac class, then using the definition of Kasparov's product one
computes for any $\alpha, \beta\in \K^0(X)$,
$
(\alpha,\beta)_\K =  \Dirac_\alpha \otimes_{C(X)} \beta =
\Ind (\Dirac_{\alpha\otimes\beta})\,.
$

If $\alg$ also satisfies Poincar\'{e} duality in bivariant local
cyclic homology, then there is a nondegenerate pairing
$
(-,-)_\H:\HE_i(\alg)\otimes_\complex\HE_{d-i}(\alg)\rightarrow
\complex
$
given by the explicit formula
\beq
(x,y)_\H=(x\otimes y^\op)\otimes_{\alg\otimes\alg^\op}\Xi
\label{HPpairing2}
\eeq
for $x\in\HE_i(\alg)$ and $y\in\HE_{d-i}(\alg)$. It is important to
note that there is another way to define a pairing on cyclic
homology. It arises when we replace $\Xi$ in eq.~(\ref{HPpairing2}) by
the image $\ch(\Delta)$ under the Chern character of the fundamental
class $\Delta$ that provides Poincar\'{e} duality in $\K$-theory. As
in the commutative case, the Todd class links the two pairings. We
then arrive at an analog of the classical isometry result that we
started this section with.

\begin{theorem}
Suppose that the noncommutative spacetime $\alg$ satisfies 
the universal coefficient theorem for local bivariant cyclic homology,
and that $\HE_\bullet(\alg)$ is a finite-dimensional vector space. If
$\alg$ has symmetric \textup{(}even-dimensional\textup{)} fundamental
classes both in $\K$-theory and in cyclic cohomology, then the
modified Chern character
$$
\ch \otimes_\alg \sqrt{{\sf Todd}(\alg)}\,:\,\K_\bullet(\alg) ~
\longrightarrow ~\HE_ {\bullet}(\alg)
$$
is an isometry with respect to the inner products \eqref{pairing} and
\eqref{HPpairing2}.
\label{isomtheorem}\end{theorem}

\begin{proof}
We have to prove the formula
\beq
\big(p\,,\,q\big)_\K=\big(\ch(p) \otimes_\alg
\sqrt{{\sf Todd}(\alg)} \,,\,\ch(q)\otimes_\alg \sqrt{{\sf
    Todd}(\alg)}~\big)_\H \ .
\label{chNCisometry1}\eeq
The left-hand side of eq.~(\ref{chNCisometry1}) is given by
$
(p,q)_\K= (p\otimes q^\op)\otimes_{\alg\otimes \alg^\op} \Delta
$.
Let us denote $$\Theta=\big(\ch(p)\otimes_\alg \sqrt{\sf
  Todd(\alg)}~\big)~\otimes~\big(\ch(q)^\op\otimes_\alg \sqrt{{\sf
    Todd}(\alg)}\,^\op\big) \ . $$ Then the right-hand side of
eq.~(\ref{chNCisometry1}) is given by
\begin{equation}
\begin{split}
\big(\ch(p)\otimes_\alg \sqrt{{\sf Todd}(\alg)}\,,\,&
\ch(q)\otimes_\alg \sqrt{{\sf Todd}(\alg)}~\big)_\H ~ = \quad
\both{\bbc}{\alg}{\alg^\op}{\Theta}{\alg}{\alg^\op}{\bbc}{\Xi}\\ & \\
& = \quad
\xy
(-10,0)*{\bbc}; 
(-10,0)*\xycircle(3.3,3.3){-}="a"; 
(20,10)*{\otimes_\alg}; 
(20,10)*\xycircle(3.3,3.3){.}="b"; 
(20,-10)*{\otimes_{\alg^\op}}; 
(20,-10)*\xycircle(4,4){.}="c"; 
(14,0)*{}="d";
"a"; "d" **\dir{-}?(.6)*\dir{>};
"d"; "b" **\dir{-}?(.6)*\dir{>};
"d"; "c" **\dir{-}?(.6)*\dir{>};
(2,8)*{\ch(p)\otimes \ch(q)^\op}; 
(45,10)*{\otimes_\alg}; 
(45,10)*\xycircle(3.3,3.3){.}="x"; 
(45,-10)*{\otimes_{\alg^\op}}; 
(45,-10)*\xycircle(4,4){.}="y"; 
(65,0)*{\bbc}; 
(65,0)*\xycircle(3.3,3.3){-}="z"; 
(51,0)*{}="t";
 "x";"t" **\dir{-}?(.6)*\dir{>};
 "t"; "z"**\dir{-}?(.6)*\dir{>};
 "y"; "t" **\dir{-}?(.6)*\dir{>};
(58,8)*{\Xi}; 
(23.9,10)*{}="bb";
(42,10)*{}="dd"; 
(24,-10)*{}="cc";
(41,-10)*{}="yy"; 
"bb" ; "dd" **\dir{-}?(.6)*\dir{>}+(-3,7)*{\sqrt{{\sf Todd}(\alg)}};
"cc" ; "yy" **\dir{-}?(.6)*\dir{>}+(-3,7)*{\sqrt{{\sf Todd}(\alg)}\,^\op};
\endxy  
\\ & \\
& = \quad \xy
(-10,0)*{\bbc}; 
(-10,0)*\xycircle(3.3,3.3){-}="a"; 
(20,10)*{\otimes_\alg}; 
(20,10)*\xycircle(3.3,3.3){.}="b"; 
(20,-10)*{\otimes_{\alg^\op}}; 
(20,-10)*\xycircle(4,4){.}="c"; 
(14,0)*{}="d";
"a"; "d" **\dir{-}?(.6)*\dir{>};
"d"; "b" **\dir{-}?(.6)*\dir{>};
"d"; "c" **\dir{-}?(.6)*\dir{>};
(2,8)*{\ch(p)\otimes \ch(q)^\op}; 
(45,10)*{\otimes_\alg}; 
(45,10)*\xycircle(3.3,3.3){.}="x1";
(70,10)*{\otimes_\alg}; 
(70,10)*\xycircle(3.3,3.3){.}="x"; 
(70,-10)*{\otimes_{\alg^\op}}; 
(70,-10)*\xycircle(4,4){.}="y"; 
(90,0)*{\bbc}; 
(90,0)*\xycircle(3.3,3.3){-}="z"; 
(76,0)*{}="t";
 "x";"t" **\dir{-}?(.6)*\dir{>};
 "t"; "z"**\dir{-}?(.6)*\dir{>};
 "y"; "t" **\dir{-}?(.6)*\dir{>};
(83,8)*{\Xi}; 
(23.9,10)*{}="bb";
(42,10)*{}="dd"; 
(24,-10)*{}="cc";
(41,-10)*{}="yy"; 
"b" ; "x1" **\dir{-}?(.6)*\dir{>}+(-3,7)*{\sqrt{{\sf Todd}(\alg)}};
"cc" ; "y" **\dir{-}?(.6)*\dir{>}+(-3,7)*{\sqrt{{\sf Todd}(\alg)}\,^\op};
"x1"; "x"**\dir{-}?(.6)*\dir{>}+(-3,7)*{1_\alg};
\endxy 
\\ & \\
& = \quad \xy
(-10,0)*{\bbc}; 
(-10,0)*\xycircle(3.3,3.3){-}="a"; 
(20,10)*{\otimes_\alg}; 
(20,10)*\xycircle(3.3,3.3){.}="b"; 
(20,-10)*{\otimes_{\alg^\op}}; 
(20,-10)*\xycircle(4,4){.}="c"; 
(14,0)*{}="d";
"a"; "d" **\dir{-}?(.6)*\dir{>};
"d"; "b" **\dir{-}?(.6)*\dir{>};
"d"; "c" **\dir{-}?(.6)*\dir{>};
(2,8)*{\ch(p)\otimes \ch(q)^\op}; 
(45,10)*{\otimes_\alg}; 
(45,10)*\xycircle(3.3,3.3){.}="x1";
(70,10)*{\otimes_\alg}; 
(70,10)*\xycircle(3.3,3.3){.}="x"; 
(70,-10)*{\otimes_{\alg^\op}}; 
(70,-10)*\xycircle(4,4){.}="y"; 
(90,0)*{\bbc}; 
(90,0)*\xycircle(3.3,3.3){-}="z"; 
(76,0)*{}="t";
 "x";"t" **\dir{-}?(.6)*\dir{>};
 "t"; "z"**\dir{-}?(.6)*\dir{>};
 "y"; "t" **\dir{-}?(.6)*\dir{>};
(83,8)*{\Xi}; 
(23.9,10)*{}="bb";
(42,10)*{}="dd"; 
(24,-10)*{}="cc";
(41,-10)*{}="yy"; 
"b" ; "x1" **\dir{-}?(.6)*\dir{>}+(-3,7)*{\sqrt{{\sf Todd}(\alg)}};
"cc" ; "y" **\dir{-}?(.6)*\dir{>}+(-3,7)*{1_{\alg^\op}};
"x1"; "x"**\dir{-}?(.6)*\dir{>}+(-3,7)*{\sqrt{{\sf Todd}(\alg)}};
\endxy 
\\ & \\
& = \quad
\xy
(-10,0)*{\bbc}; 
(-10,0)*\xycircle(3.3,3.3){-}="a"; 
(20,10)*{\otimes_\alg}; 
(20,10)*\xycircle(3.3,3.3){.}="b"; 
(20,-10)*{\otimes_{\alg^\op}}; 
(20,-10)*\xycircle(4,4){.}="c"; 
(14,0)*{}="d";
"a"; "d" **\dir{-}?(.6)*\dir{>};
"d"; "b" **\dir{-}?(.6)*\dir{>};
"d"; "c" **\dir{-}?(.6)*\dir{>};
(2,8)*{\ch(p)\otimes \ch(q)^\op}; 
(45,10)*{\otimes_\alg}; 
(45,10)*\xycircle(3.3,3.3){.}="x"; 
(45,-10)*{\otimes_{\alg^\op}}; 
(45,-10)*\xycircle(4,4){.}="y"; 
(65,0)*{\bbc}; 
(65,0)*\xycircle(3.3,3.3){-}="z"; 
(51,0)*{}="t";
 "x";"t" **\dir{-}?(.6)*\dir{>};
 "t"; "z"**\dir{-}?(.6)*\dir{>};
 "y"; "t" **\dir{-}?(.6)*\dir{>};
(58,8)*{\Xi}; 
(23.9,10)*{}="bb";
(42,10)*{}="dd"; 
(24,-10)*{}="cc";
(41,-10)*{}="yy"; 
"bb" ; "dd" **\dir{-}?(.6)*\dir{>}+(-3,7)*{{\sf Todd}(\alg)};
"cc" ; "yy" **\dir{-}?(.6)*\dir{>}+(-3,7)*{1_{\alg^\op}};
\endxy  \\ & \\
& = ~ \big(\ch(p)\otimes
\ch(q)^\op\big)\otimes_{\alg\otimes\alg^\op}\big({\sf
  Todd}(\alg)\otimes_\alg \Xi\big) \ .
\end{split}
\label{LHSdiagcomp}\end{equation}
The step from the third to the fourth line of eq.~(\ref{LHSdiagcomp})
uses symmetry of the fundamental class~$\Xi$.

We now use the definition of the Todd class to compute 
\begin{equation}
\begin{split}
{\sf Todd}(\alg)\otimes_\alg \Xi~ & =~ \big(\Xi^\vee \otimes_{\alg^\op}
\ch(\Delta)\big)\otimes_\alg \Xi\\ & \\
& = \quad
\xy
(0,0)*{\bbc}; 
(0,0)*\xycircle(3.3,3.3){-}="a"; 
(20,10)*{\alg}; 
(20,10)*\xycircle(3.3,3.3){-}="b"; 
(20,-10)*{\alg^\op}; 
(20,-10)*\xycircle(3.3,3.3){-}="c"; 
(14,0)*{}="d";
"a"; "d" **\dir{-}?(.6)*\dir{>};
"d"; "b" **\dir{-}?(.6)*\dir{>};
"d"; "c" **\dir{-}?(.6)*\dir{>};
(7,8)*{\Xi^\vee}; 
(35,10)*{\alg}; 
(35,10)*\xycircle(3.3,3.3){-}="x"; 
(35,-10)*{\alg^\op}; 
(35,-10)*\xycircle(3.3,3.3){-}="y"; 
(55,0)*{\bbc}; 
(55,0)*\xycircle(3.3,3.3){-}="z"; 
(41,0)*{}="t";
 "x";"t" **\dir{-}?(.6)*\dir{>};
 "t"; "z"**\dir{-}?(.6)*\dir{>};
 "y"; "t" **\dir{-}?(.6)*\dir{>};
(48,8)*{\ch(\Delta)}; 
(22.7,-10)*{}="cc";
(32.5,-10)*{}="yy"; 
"cc" ; "yy" **\dir{~};
(75,10)*{\alg}; 
(75,10)*\xycircle(3.3,3.3){-}="a1"; 
(75,-10)*{\alg^\op}; 
(75,-10)*\xycircle(3.3,3.3){-}="b1"; 
(95,0)*{\bbc}; 
(95,0)*\xycircle(3.3,3.3){-}="c1"; 
(81,0)*{}="d1";
 "a1";"d1" **\dir{-}?(.6)*\dir{>};
 "d1"; "c1"**\dir{-}?(.6)*\dir{>};
 "b1"; "d1" **\dir{-}?(.6)*\dir{>};
(88,8)*{\Xi}; 
"b"; "a1" **\crv{(42,23)};
\endxy \\ & \\
& = ~ \big(\Xi^\vee\otimes_\alg \Xi\big)\otimes_{\alg^\op} \ch(\Delta) 
~=~ \ch(\Delta) \ .
\end{split}
\label{Toddcomp}\end{equation}
Inserting eq.~(\ref{Toddcomp}) into eq.~(\ref{LHSdiagcomp}) we arrive
at
\[
\begin{split}
\big(\ch(p)\otimes_\alg \sqrt{{\sf Todd}(\alg)}\,,\,
\ch(q)\otimes_\alg \sqrt{{\sf Todd}(\alg)}~\big)_\H & =
 \big(\ch(p)\otimes\ch(q^\op)\big)\otimes_{\alg\otimes\alg^\op}\ch(\Delta) \\[4pt]
& = \ch\big((p\otimes q^\op)\otimes_{\alg\otimes\alg^\op}\Delta\big) \\[4pt]
& = \ch\big((p,q)_\K\big) \ .
\end{split}
\]
The result now follows from the fact that the Chern character  $\bbz=
\KK_0(\bbc, \bbc) \stackrel{\ch}{\longrightarrow} \HE_0(\bbc,\bbc) =
\bbc$ is injective. 
\end{proof}

This result has the following physical consequence, in which the
$\K$-orientation requirement provides the natural generalization of
the Freed-Witten anomaly cancellation condition~\cite{FW} to our
noncommutative settings. 
\begin{cor}
Let $\alg,\dalg$ be noncommutative D-branes such that $\alg$ satisfies
the hypotheses of Theorem~\ref{isomtheorem}, and with given
$\K$-oriented morphism $f\colon \alg \to \dalg$ and Chan-Paton bundle
$\xi \in \K_\bullet(\dalg)$. Then there is a noncommutative analogue
of the Minasian-Moore formula for the Ramond-Ramond charge, valued in
$\HE_\bullet(\alg)$ and given by
$$
{\sf Q}(\dalg,\xi) = \ch\big(f_!(\xi)\big)  \otimes_\alg \sqrt{{\sf
    Todd}(\alg)} \ .
$$
\label{NCchargecor}\end{cor}

\begin{remark}
If $\dalg$ is also a PD algebra, then by using the
Grothendieck-Riemann-Roch formulas of~\S\ref{NCGRR} we may write
the noncommutative charge vector in the form
$$
{\sf Q}(\dalg,\xi) = \ch(\xi)  \otimes_\dalg
\,\Todd(\dalg)\otimes_\dalg (f*)
\otimes_\alg\sqrt{{\sf Todd}(\alg)}\,^{-1} \ .
$$
If the $\HE$-theory Gysin class $f*\in\HE(\dalg,\alg)$ induced by the
morphism $f:\alg\to\dalg$ obeys $$(f*)\otimes_\alg\sqrt{{\sf
    Todd}(\alg)}\,^{-1}=\Lambda\otimes_\dalg(f*)$$ for some element
$\Lambda\in\HE(\dalg,\dalg)$, then one has a noncommutative version of
the Wess-Zumino class (\ref{WZclassical}) valued in
$\HE_\bullet(\dalg)$, intrinsic to the D-brane $(\dalg,\xi)$ itself,
and given by
$$
{\sf D}_{\rm WZ}(\dalg,\xi)=\ch(\xi)  \otimes_\dalg
\,\Todd(\dalg)\otimes_\dalg\Lambda \ .
$$
This formula enables one to derive a criterion under which the
noncommutative charge vector is covariant under T-duality
transformations. Suppose that the D-branes $(\dalg,\xi)$ are
$(\dalg',\xi'\,)$ are $\KK$-equivalent, where the equivalence is
implemented by inverse classes $\alpha\in\KK(\dalg,\dalg'\,)$ and
$\alpha^{-1}\in\KK(\dalg',\dalg)$ with
$\xi'=\xi\otimes_\dalg\alpha$. Then the Todd classes are related
by~\cite[Theorem~7.4]{BMRS}
\beq
\Todd(\dalg'\,)=\ch(\alpha)^{-1}\otimes_\dalg\,\Todd(\dalg)\otimes_\dalg
\,\ch(\alpha) \ .
\label{ToddKK}\eeq
If
$$
\Lambda'=\ch(\alpha)^{-1}\otimes_\dalg\Lambda\otimes_\dalg
\,\ch(\alpha)
$$
then the Wess-Zumino classes are related covariantly as
$$
{\sf D}_{\rm WZ}(\dalg',\xi'\,)={\sf D}_{\rm WZ}(\dalg,\xi)\otimes_\dalg
\,\ch(\alpha) \ .
$$
Suppose that in an analogous way $(\dalg',\xi'\,)$ is $\KK$-equivalent
to $(\dalg'',\xi''\,)$, where the equivalence is implemented by
inverse classes $\alpha'\in\KK(\dalg',\dalg''\,)$ and
$\alpha'{}^{-1}\in\KK(\dalg'',\dalg'\,)$ with
$\xi''=\xi'\otimes_{\dalg'}\alpha'$, such that $\dalg$, $\dalg''$ are
Morita equivalent with associated $\KK$-equivalence
$\alpha\otimes_{\dalg'}\alpha'$. Then by multiplicativity of the Chern
character one has
$$
{\sf D}_{\rm WZ}(\dalg'',\xi''\,)={\sf D}_{\rm WZ}(\dalg,\xi)\otimes_\dalg
\,\ch(\alpha\otimes_{\dalg'}\alpha'\,)
$$
and so the Wess-Zumino classes are related by the image of this
$\KK$-equivalence in
$$\HE_\bullet(\dalg''\,)\cong\HE_\bullet(\dalg)\ . $$ These formulas all
express the T-duality invariance of the noncommutative D-brane
charge formula. The stated conditions above presumably reflect the
need for addition of Myers terms to the Wess-Zumino class, through
appropriate modification of the morphism $f:\alg\to\dalg$. 
\end{remark}

\subsection{Twisted D-branes\label{TwistedD}}

A nice application of this formalism is to the example of D-branes in
a compact, oriented even-dimensional manifold $X$ with background
$H$-flux. Let $$\alg_H~:=~ {\sf CT}(X, H)~=~C_0(X,\bun_H)$$ be the
$C^*$-algebra of sections of a locally trivial bundle of compact
operators $\bun_H\to X$ with Dixmier-Douady invariant $[H] \in \H^3(X,
\bbZ)$, whose image in real cohomology is represented by a closed
three-form $H$ on $X$. The opposite algebra is
$\alg_H^\op=\CT(X,-H)$. Recall that $$ {\sf CT}(X, H) =
C_0(P_H,\cK)^{{PU}}$$ consists of ${PU}$-equivariant $\cK$-valued
continuous maps from a principal ${PU}$-bundle $P_H\to X$ with
Dixmier-Douady class $[H]$, where the projective unitary group
$PU=PU(\hil)$ of a fixed separable Hilbert space $\hil$ acts by the
adjoint action on the algebra of compact operators $\cK=\cK(\hil)$. Let
$\cL^1_{P_H} = P_H\times_{{PU}} \cL^1$ be the algebra bundle of trace
class operators $\cL^1=\cL^1(\hil)$, where $PU$ acts by the adjoint
action on $\cL^1$.  Let $\alg_H^\infty:={\sf CT}^\infty (X, H) =
C^\infty(X, \cL^1_{P_H})$ denote the locally convex $*$-algebra of
smooth sections, which is a smooth subalgebra of $\alg_H$. Then the
inclusion map $\iota : {\sf CT}^\infty (X, H)\hookrightarrow {\sf CT}
(X, H)$ induces an isomorphism on $\K$-theory~\cite{MS2} $$\iota^!
\,:\, \K_\bullet\big({\sf CT}^\infty
(X,H)\big)~\stackrel{\approx}{\longrightarrow}~\K_\bullet\big({\sf CT}
(X,H)\big) \ . $$

Upon choosing a bundle gerbe connection~\cite{Murray}, there are
natural isomorphisms~\cite{BCMMS}
$$
\K_\bullet \big( {\sf CT}^\infty(X, H)\big) \cong \K^\bullet(X, H)
$$
where the right-hand side is the twisted $\K$-theory of
$X$ expressed in terms of projective Hilbert bundles with fixed
reduction of structure group to unitaries of the form $1 + \text{trace
  class operators}$. The same data also defines an
isomorphism~\cite{MS2}
$$
\HP_\bullet \big(  {\sf CT}^\infty(X, H)\big) \cong \H^\bullet (X, H)
$$
where the right-hand side denotes the periodised twisted de Rham
cohomology $$\H^\bullet(X, H) = \H^\bullet\big(\Omega^\#(X)\,,\, \dd -
H\wedge \big) \ . $$ 
The canonical Connes-Chern character
$$
{\ch} \,:\,  \K_\bullet \big( {\sf CT}^\infty(X, H)\big) ~
\longrightarrow~ \HP_\bullet \big( {\sf CT}^\infty(X, H)\big)
$$
can be expressed in terms of differential forms on
$X$~\cite{BCMMS,MS,MS2}, leading to the twisted Chern character
$$
{\ch}_H \,:\, \K^\bullet(X, H) ~\longrightarrow~ \H^\bullet (X, H) \ .
$$
This can be viewed as the Chern-Weil representative
of the Connes-Chern character on twisted $\K$-theory.

The algebra $\alg_H$ is not generally a PD algebra, but it
is a member of the class $\underline{\mathfrak{D}}$ and a dual algebra
$\tilde\alg_H$ may be constructed as follows~\cite{BMRS,Tu}. Let $\Cl(X)$
be the Clifford algebra bundle of the cotangent bundle of $X$ with
respect to a chosen riemannian metric. Its Dixmier-Douady invariant
is the third integral Stiefel-Whitney class
$\W_3(X)\in\H^3(X,\zed)$, which is the topological obstruction
to the existence of a spin$^c$ structure on $X$. The Dirac operator 
class $[\Dirac]$, constructed in the usual fashion from the riemannian
metric and the Levi-Civita connection on $TX$,
naturally lives in $\KK(C_0(X,\Cl(X)),\complex)$ and
gives a fundamental class $\Dirac\otimes\Dirac$. Kasparov product with
$\Dirac\otimes\Dirac$ establishes Poincar\'e duality
between the algebras $C(X)$ and $C_0(X,\Cl(X))$. Set
$\tilde\alg_H:=\CT(X,\W_3(X)-H)\cong
C_0(X,\bun_{-H}\otimes\Cl(X))$. Let $p_i:X\times X\to X$, $i=1,2$ be
the projection onto the $i$-th factor, and $\delta:X\to X\times X$
the diagonal map. Then $$\alg_H\otimes\tilde\alg_H\cong\CT\big(X\times
X\,,\,p_1^*(H)+p_2^*(\W_3(X)-H)\big)$$ with
$\delta^*(\alg_H\otimes\tilde\alg_H)\cong\CT(X,\W_3(X))\otimes
\mathcal{K}$, which is Morita equivalent to $C_0(X,\Cl(X))$. 
It follows that a fundamental class for the pair
$(\alg_H,\tilde\alg_H)$ is
$$
\Delta=\big[\delta^*\big]_{\KK}\otimes_{\CT(X,\W_3(X))}\Dirac \ .
$$
When $X$ is spin$^c$, one has $\W_3(X)=0$ and the algebra
$C_0(X,\Cl(X))$ is Morita equivalent to $C(X)$~\cite{Plymen}. Moreover,
in this case $\tilde\alg_H=\alg_H^\op$ and the restriction of
$\alg_H\otimes\alg_H^\op$ to the diagonal is stably isomorphic to
$C(X)$. 

Since $\H^\bullet(X,\W_3(X))\cong\H^\bullet(X,\real)$, a fundamental
cyclic cohomology class $\Xi$ for $\alg_H$ is provided as usual by
the homology orientation cycle $[X]$. In the spin case, the Todd class ${\sf
Todd}(\alg_H)$ is given by the Atiyah-Hirzebruch class
$\widehat{A}(X)$ of the manifold $X$ and the {modified Chern
character} reduces to
$$
\Ch_H(\xi) = \ch_H(\xi) \wedge  \sqrt{\widehat{A}(X)}
$$
for $\xi\in\K^\bullet(X,H)$. In general, the modified
Chern character gives an isometry between the natural
bilinear pairings on twisted $\K$-theory and cohomology which are
defined as follows. The Grothendieck group of the category of
$\Cl(X)$-modules is isomorphic to the twisted $\K$-theory
$\K^0(X,\W_3(X))\cong\K_0(C_0(X,\Cl(X)))$. Given any $\Cl(X)$-module $\cC$
over $X$ and choosing a Clifford connection on $\cC$, there is a Dirac
operator $\Dirac_\cC$ acting on $C_0(X,\cC)$. The index map
$[\cC]\mapsto \Ind(\Dirac_\cC)$ defines a
homomorphism $\K^0(X,\W_3(X))\to\zed$. This map can be
computed via the standard Atiyah-Singer
index theorem~\cite[Theorem~4.3]{BerVer}
$$
\Ind(\Dirac_\cC)=\int_X\,\ch(\cC)\wedge\widehat{A}(X) \ .
$$
The tensor product of twisted
bundles then defines a bilinear Poincar\'e pairing
$$
\K^\bullet(X,H)\otimes\K^\bullet\big(X\,,\,\W_3(X)-H\big)~
\longrightarrow~\K^0\big(X\,,\,\W_3(X)\big)~\xrightarrow{\Ind}~\zed \
.
$$
Using the Poincar\'e pairing
(\ref{cohpairingclass}) with the cup product on twisted cohomology
$$\H^\bullet(X,H)\otimes\H^\bullet\big(X\,,\,\W_3(X)-H\big)~
\longrightarrow~\H^\even(X,\real) \ , $$
the isometric pairing formula follows.

Let us now consider a D-brane $(W, E, f)$ in $X$, where $f:
W\hookrightarrow X$ is the inclusion of the oriented submanifold $W$,
and assume for simplicity that $X$ is spin (as is the case when $X$ is
the spacetime of Type~II superstring theory). In this case, the
constructions of~\S\ref{Koriented} determine a canonical element
$f!\in\KK_d(\CT(W,f^*H+W_3(N_XW)),\CT(X,H))$ where $d=\dim(X)-\dim(W)
\mod 2$ (see~\cite{CareyWang} for details). Since $X$ is spin and $W$
is oriented, one has $\W_3(N_XW)=\W_3(W)$. For the D-brane algebra we
thus take
$$
\dalg=\CT\big(W\,,\,f^*H+\W_3(W)\big)
$$
and the Chan-Paton bundle $E$ is generically an infinite dimensional
twisted bundle \cite{BCMMS}
whose class lives in the twisted $\K$-theory
$\K^0(W,f^*H+\W_3(W))$. There are two special cases of interest. When
the brane worldvolume is spin$^c$, {\it i.e.}, $\W_3(W)=0$, the
algebra $\dalg$ is just the restriction of $\alg_H$ to the submanifold
$W$. The Chan-Paton bundle on the brane in this case is 
an infinite dimensional twisted 
bundle $E\in\K^0(W,f^*H)$. When $H$ is a torsion element in
$\H^3(X,\zed)$, the algebra $\dalg$ is an Azumaya algebra, and the
bundle $E$ is of finite rank and represents
't~Hooft flux corresponding to a finite number of D-branes wrapping
$W$~\cite{Kapustin,WOverview}. (The torsion condition guarantees
anomaly cancellation on spacetime-filling brane-antibrane pairs.) But
generically it is of infinite rank, corresponding to an infinite number
of wrapped branes. Alternatively, we can use honest Chan-Paton vector
bundles $E\in\K^0(W)$ and (stably) commutative algebras on the
D-branes by replacing the spin$^c$ condition with the Freed-Witten
anomaly cancellation condition~\cite{FW}
$$
f^*H+\W_3(W)=0
$$
in $\H^3(W,\zed)$, which guarantees the absence of global worldsheet
anomalies for Type~II superstrings. In any case, we obtain the twisted
D-brane charge vector
$$
{\sf Q}_H(W, E) = \ch_H\big(f_!(E)\big) \wedge \sqrt{\widehat{A}(X)}
$$
in $\H^\bullet(X,H)$. Only Freed-Witten anomaly-free D-branes have the
same Wess-Zumino class (\ref{WZclassical}) valued in ordinary
cohomology $\H^\bullet(W,\rat)$ as in the untwisted case.

\subsection{Noncommutative D2-branes}

D2-branes can wrap two-dimensional manifolds, which in the presence of
a constant $B$-field are described by noncommutative Riemann surfaces.
Let $\Gamma_g$ be the fundamental group of a compact, oriented
Riemann surface $\Sigma_g$ of genus $g\ge 1$ with the presentation
$$
\Gamma_g = \Bigl\{\mbox{$U_j, V_j\,,\, j=1, \ldots, g~
\Big|~\prod\limits_{j=1}^g\,[U_j, V_j] = 1$}\Bigr\} \ .
$$
Since $\H^2(\G_g, { U}(1)) \cong
\real/\zed$, for each $\theta\in [0, 1)$
there is a unique multiplier $\sigma_\theta$ (a two-cocycle with
values in $U(1)$) on $\G_g$ up to isomorphism, representing the
holonomy of a non-trivial $B$-field on $\Sigma_g$. Let
$\alg_\theta:=C^*_r(\Gamma_g, \sigma_\theta)$ be the $C^*$-algebra
generated by the unitaries $U_j$ and $V_j$ with the relation
$$
\prod_{j=1}^g \,[U_j, V_j] = \exp (2\pi \ii\theta) \ .
$$
In~\cite{CHMM,CM} the $\K$-theory of the algebra $\alg_\theta$ is
computed as follows:
\begin{itemize}
\item 
$\K_0(\alg_\theta) \cong\K^0(\Sigma_g)= \mathbb{Z}^2$, and if $\theta\not\in
\rat$ the algebras $\twga$ are distinguished by the image of the
canonical trace $\tr:C_r^*(\Gamma_g,\sigma_\theta)\to\complex$,
induced by evaluation at the identity element of $\Gamma_g$, on
$\K$-theory as
$$
\tr^!(n\,e_0 + m\,e_1) = n + m\,\theta \ ,
$$
with the $\K$-theory generators chosen as $e_0 = [1]$, the class of the
identity element, and $e_1$ such that $\tr^!(e_1) = \theta$; and
\item  $\K_1(\alg_\theta) \cong \K^1(\Sigma_g)=\mathbb{Z}^{2g}$, with $U_j$
  and $V_j$ forming a basis for $\K_1(\alg_\theta)$.
\end{itemize}

The algebra $\alg_\theta$ is a PD algebra, and the inverse of the
fundamental class of $\alg_\theta$ is the Bott class
\beq 
\Delta^\vee= e^{\phantom{\op}}_0\otimes e^\op_1 - e^{\phantom{\op}}_1
\otimes e^\op_0 + \sum_{j=1}^g\,\left(U_j^{\phantom{\op}}
\otimes V^\op_j - V_j^{\phantom{\op}}\otimes U^\op_j\right) \ .
\nonumber \eeq
Let
$$
\mu_\theta \,:\, \K^\bullet(\Sigma_g) ~ \stackrel{\approx}
{\longrightarrow}~ \K_{\bullet}\big(C^*_r(\Gamma_g, \sigma_\theta)\big)
$$
be the twisted version of the Kasparov isomorphism~\cite{CHMM} for the
Lie group $SO_0(2,1)\supset\Gamma_g$ with $\Gamma_g\setminus
SO_0(2,1)/SO(2)\cong\Sigma_g$, and let $\nu_\theta$ be its analog in
periodic cyclic homology. Then there is a commutative diagram
\beq
\begin{CD}
\K^\bullet(\Sigma_g)  @>\mu_\theta>>
\K_{\bullet}(\alg_\theta)   \\         
      @V{\ch}VV          @VV{\ch_{\Gamma_g}} V     \\
\H^\bullet  (\Sigma_g,\zed)    @>>\nu_\theta>  \HE_{\bullet}
(\alg_\theta)  
\end{CD}
\label{Sigmadiag}\eeq
whose maps are all isomorphisms. Using the diagram (\ref{Sigmadiag})
one shows that the Todd class $\Todd(\alg_\theta)$ is
determined by $\nu_\theta(\Todd(\Sigma_g))$.

The modified {Chern character} in this case is $\Ch_\theta:
\K_{\bullet}(\alg_\theta) \to \HE_{\bullet}
(\alg_\theta)$, where
$$
\Ch_\theta(\zeta) =\nu_\theta\Big( \ch\big(\mu_\theta^{-1}
(\zeta)\big) \smile \sqrt{{\Todd}(\Sigma_g)}~\Big)
$$
for $\zeta\in\K_{\bullet}(\alg_\theta)$ gives an {isometry}
between the natural bilinear pairings in $\K$-theory and in
cohomology, lifted from those of $\Sigma_g$ via the diagram
(\ref{Sigmadiag}). This is the charge vector for spacetime-filling
noncommutative D2-branes on the Riemann surface $\Sigma_g$. More
generally, the {charge vector} for a D-brane $(\dalg, \xi)$, with
$\K$-oriented morphism $f: \alg_\theta\to \dalg$ and Chan-Paton
bundle $\xi \in \K_\bullet(\dalg)$, is given by
$$
\quad{\sf Q}_\theta(\dalg, \xi) = \nu_\theta\Big(
\ch\big(\mu_\theta^{-1}\circ
f_!(\xi)\big) \smile \sqrt{{\Todd}(\Sigma_g)}~\Big) \ .
$$

This example also furnishes an explicit computation of the Todd class
for the noncommutative manifold $C^*(\Gamma_g)$. Since the Baum-Connes
conjecture holds for $\Gamma_g$ and $\Gamma_g$ has the
Dirac-dual~Dirac property, it follows that the algebras $C(\Sigma_g)$
and $C^*(\Gamma_g)$ are $\KK$-equivalent. Explicitly,
$(C(\Sigma_g),C^*(\Gamma_g))$ is a PD pair whose natural inverse
fundamental class $\Delta_{\Gamma_g}^\vee=[F\Gamma_g]$ is the class of
the universal flat $C^*(\Gamma_g)$-bundle in
$\KK(\complex,C(\Sigma_g)\otimes C^*(\Gamma_g))$. The corresponding
twisting of the class of the Dirac operator $\dirac$ on $\Sigma_g$
defines the Baum-Connes assembly map
$$
\alpha=[F\Gamma_g]\otimes_{C(\Sigma_g)}\dirac
$$
in $\K$-theory, which is an invertible element in
$\KK(C(\Sigma_g),C^*(\Gamma_g))$. Let $\alpha^{-1}$ be the dual
Dirac class in $\KK(C^*(\Gamma_g),C(\Sigma_g))$. Then the fundamental
$\K$-homology class $\Delta_{\Gamma_g}\in\KK(C(\Sigma_g)\otimes
C^*(\Gamma_g),\complex)$ is constructed as
$$
\Delta_{\Gamma_g}=\alpha^{-1}\otimes_{C(\Sigma_g)}(\dirac\otimes\dirac)
\ .
$$
The natural choices of fundamental classes in cyclic cohomology are
constructed as follows. Consider the usual fundamental homology class
of $\Sigma_g$, viewed as an element
$[\Sigma_g]\in\HE(C(\Sigma_g),\complex)$, and the canonical trace
$\tr$ on $C^*(\Gamma_g)$, viewed as an element of
$\HE(C^*(\Gamma_g),\complex)$. Then
$$
\Xi_{\Gamma_g}=[\Sigma_g]\otimes[\tr]_\HE \ .
$$

A short calculation then shows that, with these natural choices, the
Todd class of the algebra $C^*(\Gamma_g)$ is given by
\beq
\Todd\big(C^*(\Gamma_g)\big)=\ch(\alpha)^{-1}\otimes_{C(\Sigma_g)}\,
\Todd\big(C(\Sigma_g)\big)\otimes_{C(\Sigma_g)}\,\ch(\alpha) \ ,
\label{ToddKKSigma}\eeq
where
$$\Todd\big(C(\Sigma_g)\big)=
\big([\Sigma_g]\otimes[\Sigma_g]\big)\otimes_{C(\Sigma_g)}
\big(\,\Todd(\Sigma_g)\otimes\Todd(\Sigma_g)\big)~\in~
\HE\big(C(\Sigma_g)\,,\,C(\Sigma_g)\big)$$ is given by cup product
with the cohomology class $\Todd(\Sigma_g)$ dual to
$\ch(\dirac)$. This is a special case of~\cite[Corollary~7.5]{BMRS}
(see eq.~(\ref{ToddKK})). There is a canonical $\K$-oriented morphism
$c:\complex\to C^*(\Gamma_g)$, and the composition product
\[
[c]_\HE \otimes_{C^*(\Gamma_g)} \Todd\big(C^*(\Gamma_g)\big) 
\otimes_{C^*(\Gamma_g)} [\tr]_\HE ~\in~ \HE(\bC,\bC) =\bC
\]
is the Todd \emph{genus} of $C^*(\Gamma_g)$~\cite[Remark~7.13]{BMRS},
which works out simply to $1-g$, the Todd genus of the Riemann surface
$\Sigma_g$.

\begin{remark}

A similar argument computes the Todd genus for the reduced
$C^*$-algebra $C_r^*(\Gamma)$, where $\Gamma$ is any discrete group
satisfying the Baum-Connes conjecture, having the Dirac-dual Dirac
property, and whose classifying space $B\Gamma$ is a spin$^c$
manifold. Then the Todd genus of $C_r^*(\Gamma)$ is equal to the Todd
genus of $B\Gamma$. We will revisit these examples in the next section
in conjunction with noncommutative correspondences.

\end{remark}

\section{Correspondences and open string
  T-duality\label{CorrTduality}}

There are several different but equivalent definitions of
$\KK$-theory, each adapted to different applications.  As we discuss
below, the one that is best suited for T-duality in open string theory
is due to Baum, Connes and Skandalis~\cite{CS} in the case of
manifolds. In this final section we will propose a new description of
elements of $\KK$-theory based on the constructions
of~\S\ref{NCRRTheorem}. This description generalizes the
Baum-Connes-Skandalis description of elements of $\KK$-theory to the
setting of generic separable $C^*$-algebras, and it gives a more
precise meaning to the point of view that Kasparov bimodules
generalize $*$-homomorphisms of $C^*$-algebras. In this sense it is
very closely related to Cuntz's picture of $\KK$-theory. In the
special case of the Kasparov groups $\KK(C_0(X), C_0(Y)\otimes \alg)$,
where $\alg$ is a unital separable $C^*$-algebra, the correspondence
description of~\cite{CS} has been extended by Block and
Weinberger~\cite{BW}.

\subsection{Geometric correspondences and $\KK$-theory for
  manifolds\label{GeomCorr}}

We will first recall the description of~\cite{CS} for the $\KK$-theory
groups $\KK(X,Y):=\KK(C_0(X),C_0(Y))$ of smooth manifolds $X,Y$. In
this picture, elements of $\KK_d(X, Y)$ are represented by
{\em geometric correspondences}
\beq
\xymatrix @=4pc { & (Z , E) \ar[dl]_{f_X} 
\ar[dr]^{g^Y} & \\
X  &   &    Y }
\label{commcorrdef}\eeq
where $E$ is a complex vector bundle over the smooth manifold $Z$, and
$f_X, g^Y$ are continuous maps from $Z$ to $X$ and $Y$, respectively,
such that $f_X$ is proper, $g^Y$ is smooth and $\K$-oriented, and $d=
\dim(Z)- \dim(Y) \mod2$. Elements of $\KK_d(X, Y)$ define
homomorphisms in $ {\rm Hom}(\K^\bullet( X), \K^{\bullet+d}(Y))$. To
each correspondence (\ref{commcorrdef}) there corresponds a morphism
of $\K$-theory groups given by the assignment
\beq
\big(Z\,,\,E\,,\,f_X\,,\,g^Y\big)~
\longmapsto~ g^Y{}_!\big(f_X{}^*(-) \otimes E\big) \ ,
 \label{corrassign}\eeq
implemented by a $\KK$-theory element $[f_X]_{\KK}\otimes_{C_0(Z)} [[E]]
\otimes_{C_0(Z)}(g^Y)!$. Here $[f_X]_{\KK}\in \KK(X,Z)$ is defined by
the pullback morphism $f_X{}^*\colon C_0(X) \to C_0(Z)$, $[[E]]\in
\KK(Z,Z)$ is the \emph{bivariant} $\K$-theory class defined by the
vector bundle $E$~\cite[\S24.5]{Black}, and $(g^Y)!$ is defined by the
$\K$-orientation as in \S\ref{Koriented}. This is reminiscent of the
smooth analog of the Fourier-Mukai transform, and as such is
well-suited to describe T-duality, as we discuss below. It is a
geometric presentation of the analytic index for  families of elliptic
operators on $X$ parametrized by $Y$.
 
Two correspondences 
$$
\xymatrix @=4pc { & (Z_j , E_j) \ar[dl]_{f_j} 
\ar[dr]^{g_j} & \\
X  &   &    Y }
$$
with $j=1, 2$ are said to be \emph{cobordant} if there is a
``correspondence with boundary''
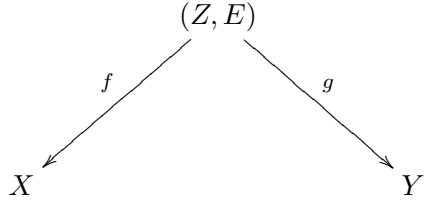
$$
\xymatrix @=4pc { & (Z , E) \ar[dl]_{f} 
\ar[dr]^{g} & \\
X  &   &    Y }
$$
such that
\begin{enumerate}
 \item $\partial Z = Z_1\sqcup Z_2$;
\item  $E|_{Z_j} = E_j$, $j=1,2$; and
\item  $f|_{Z_j} = f_j$, $j=1,2$ and $g|_{Z_j} =
  (-1)^{j+1}\,g_j$, $j=1,2$, where the minus sign indicates the same map
  with the opposite $\K$-orientation.
 \end{enumerate}
We denote the collection of all cobordism classes of
correspondences between $X$ and $Y$ by $\Omega(X,Y)$. Then the
assignment (\ref{corrassign}) descends to a well-defined surjection
\beq
\Omega_{d\text{ mod }2}(X,Y)~\longrightarrow~\KK_d(X,Y) \ ,
\label{cobordsurject}\eeq
and hence every element of $\KK_d(X,Y)$ can be represented by a
cobordism class of geometric correspondences between $X$ and $Y$. The
image under the map (\ref{cobordsurject}) of the cobordism class
represented by a correspondence (\ref{commcorrdef}) is denoted
$[Z,E]$, with the summation rules $$[Z,E_1\oplus
E_2]~=~[Z,E_1]+[Z,E_2] \qquad \mbox{and} \qquad [Z_1\sqcup
Z_2,E_1\sqcup E_2]~=~[Z_1,E_1]+[Z_2,E_2]\ . $$ Analytically, these are
``cobordism classes'' of families of elliptic operators on $X$ that
are parametrized by $Y$.
 
The special case $\KK(X, \pt)$, with $X$ compact,
is the $\K$-homology of $X$ which can be represented
analytically by ``cobordism'' classes of generalized elliptic operators
on $X$. Since a $\K$-oriented map $Z\to\pt$ is the same thing as a
spin$^c$ structure on $Z$, a correspondence in this case is simply a spin$^c$
manifold $Z$ equipped with a complex vector bundle $E$ and a proper
map to $X$. Since $X$ is compact, properness just
means $Z$ is compact. Thus a correspondence in this case is just
a geometric $\K$-cycle $(Z,E,f_X)$ over $X$~\cite{BD}. The
kernel of the map (\ref{cobordsurject}) is then the subgroup generated by the
equivalence relation on geometric $\K$-cycles given by Baum-Douglas vector
bundle modification \cite{BD}. On the other hand, $\KK(\pt, Y)$ is just the
$\K$-theory of $Y$ represented via an ABS-type construction using a
spin$^c$ structure on the bundle $TZ\oplus(g^Y)^*(TY)$ rather than on
$TZ$. Both of these limiting cases naturally define the charge of a
D-brane $(Z,E)$, supported on $Z$ with Chan-Paton bundle $E$, in
spacetime $X$ or $Y$.

In the general case, the kernel of the map (\ref{cobordsurject}) 
was described in~\cite{BB}. Presumably it could also
be described using Jakob's approach~\cite{Jakob1,Jakob2} by replacing
actual vector bundles $E\to Z$ with $\K$-theory classes,
\emph{i.e.}, by considering instead the \emph{virtual correspondences}
\beq
\xymatrix @=4pc { & (Z , \xi) \ar[dl]_{f_X} 
\ar[dr]^{g^Y} & \\
X  &   &    Y }
\label{commcorrvirt}\eeq
where now $\xi\in\K^\bullet(Z)$. Given a virtual correspondence
(\ref{commcorrvirt}), let $F\to Z$ be a smooth, real spin$^c$
vector bundle of even rank. Let $\id^\bbR_Z$ denote the trivial real
line bundle over $Z$. Let $\pi:S\to Z$ be the unit sphere bundle of
$F\oplus\id^\bbR_Z$. It is also a spin$^c$ bundle with
even-dimensional fibres which has a nowhere zero
section $s:Z\to S$ defined by $z\mapsto(0_z,1)$. Both $\pi$ and $s$
are proper maps endowed with a canonical $\K$-orientation. Then the
\emph{vector bundle modification} of (\ref{commcorrvirt}) is defined
to be the virtual correspondence
$$
\xymatrix @=4pc { & \big(S \,,\, s_!(\xi)\big) \ar[dl]_{f_X\circ\pi} 
\ar[dr]^{g^Y\circ\pi} & \\
X  &   &    Y \ . }
$$
The map (\ref{cobordsurject}) is then made into an isomorphism by
quotienting by the subgroup of cobordism classes of virtual
correspondences under the equivalence relation generated by this
generalized notion of vector bundle modification. This stabilization
makes the assignment (\ref{corrassign}) into a well-defined natural
transformation of bivariant theories which is an
equivalence when $X,Y$ are compact~\cite{BB}. Note that the grading
on $\KK_d(X,Y)$ in this case is given on homogeneous correspondences
with $\xi\in\K^j(Z)$ by $d=\dim(Z)-\dim(Y)+j \mod 2$.

The composition product of $\KK$-theory, which is notoriously
difficult to define within the analytic framework, is particularly
easy to describe in this setting. It is a morphism
$$
 \KK(X, M) \times \KK(M, Y) ~\longrightarrow~ \KK(X,Y)
$$
described as the composition of correspondences, with some 
caveats that we explain below. Given two correspondences
$$
\xymatrix @=4pc { & (Z_1 , E_1) \ar[dl]_{f_X} 
\ar[dr]^{g^M} & & (Z_2,E_2) \ar[dl]_{f_M} \ar[dr]^{g^Y}& \\
X  &   &    M & & Y }
$$
we define their {\em composition} to be the correspondence
(\ref{commcorrdef}) with the fibred product $Z= Z_1 \times_M Z_2$, the
bundle $E = E_1\boxtimes E_2$, and the compositions $f_X:Z\to Z_1\to
X$, $g^Y:Z\to Z_2\to Y$. The {\em caveat} is that in order
for the space $Z$ to be a manifold, $f_M$ should be smooth and
the maps $g^M$ and $f_M$ have to
be transverse, \emph{i.e.}, for all $(z_1,z_2) \in Z$ one has
$$
\dd f_M(T_{z_2} Z_2) + \dd g^M(T_{z_1} Z_1) = T_{f_M(z_2)} M \ .
$$
Such choices can always be made by homotopy invariance of the pullback
and Gysin maps, along with standard transversality theorems. The
notation for the composition product is $$[Z_1, E_1] \otimes_M [Z_2,
E_2] = [Z,E]\ . $$ It is manifestly associative. The unit element of
the ring $\KK(X,X)$ is denoted $1_X=[\Id_X]_{\KK}$.
 
\begin{definition}
Two manifolds $X, Y$ are said to be \emph{$\KK$-equivalent} if there are
elements $$\alpha ~\in~ \KK(X, Y) \qquad \mbox{and} \qquad \beta ~\in~
\KK(Y, X)$$ such that $\alpha \otimes_Y \beta = 1_X \in \KK(X,
X)$ and $\beta \otimes_X\alpha = 1_Y \in \KK(Y, Y)$.
\end{definition}  
   
\begin{example} \label{FMtransf}
The smooth analog of the  {\em Fourier-Mukai transform} is a
correspondence
$$
\xymatrix @=4pc { & \big(M \times {\mathbb T}^n \times
  \widehat{\mathbb T}^n  \,,\, \mathcal P\big) \ar[dl]_{p_1} 
\ar[dr]^{p_2} & \\
M \times {\mathbb T}^n &   &    M \times  \widehat{\mathbb T}^n }
$$
where ${\mathbb T}^n$ is an $n$-dimensional torus, $ \widehat
{\mathbb T}^n$ is the dual torus, and $\mathcal P$ is the Poincar\'e
line bundle defined initially on $ {\mathbb T}^n \times \widehat
{\mathbb T}^n$ and then pulled back via the projection map to the
product manifold $M\times  {\mathbb T}^n \times \widehat {\mathbb
  T}^n$. It defines an element  $\alpha$ in $\KK_n(M \times {\mathbb
  T}^n, M \times  \widehat{\mathbb T}^n)$ which is a
$\KK$-equivalence. The element $\alpha$ is interpreted analytically as
the families Dirac operator, and its inverse $\beta$ is the ``dual
Dirac'' or Bott element. Viewed in this way, it is a refinement of the
usual smooth Fourier-Mukai transform. In open string theory, this is
the statement that (topological) T-duality is naturally an invertible
element of the $\KK$-theory of the dual pair of spacetimes, and is the
starting point for a general axiomatic description of T-duality for
$C^*$-algebras~\cite{BMRS}.
\end{example} 

\subsection{Noncommutative correspondences and $\KK$-theory for
  $\mbf{C^*}$-algebras\label{NCCorr}} 

We will now give a new description of $\KK$-theory via ``noncommutative
correspondences'', which are well suited for describing the
composition product. Let $\alg, \balg$ be separable $C^*$-algebras. We
will represent elements of $\KK(\alg, \balg)$ by \emph{noncommutative
  correspondences}
\beq
\xymatrix @=4pc { & (\calg , \xi) 
& \\
\alg\ar[ur]^{f^\alg}  &   &    \balg\ar[ul]_{g_\balg} }
\label{algcorrdef}\eeq
where $\xi \in \KK(\calg, \calg)$ with $\calg$ a separable
$C^*$-algebra, and $f^\alg, g_\balg$ are homomorphisms 
$\alg\to\calg$ and $\balg\to\calg$, respectively, such that $g_\balg$
is $\K$-oriented. To see that this data defines an element in
$\KK(\alg, \balg)$, note that $[f^\alg]_{\KK} \in \KK(\alg, \calg)$ and the
Gysin element corresponding to $g_\balg$ as defined
in~\S\ref{Koriented} is $g_\balg! \in \KK(\calg, \balg)$. Thus the
composition product gives an assignment
\beq
\big(\calg\,,\,\xi\,,\,f^\alg\,,\,g_\balg\big)~\longmapsto~
\big[f^\alg\big]_{\KK}\otimes_\calg \xi \otimes_\calg g_\balg! \in
\KK(\alg, \balg) \ .
\label{algcorrassign}\eeq
The associated element in $ {\rm Hom}(\K_\bullet( \alg), 
\K_{\bullet}(\balg))$ is $g_\balg{}^!(f^\alg{}_*(-) \otimes_\calg
\xi)$. Note the analogy with eq.~\eqref{corrassign}.
In the commutative case $\calg=C_0(Z)$ of eq.~(\ref{commcorrdef}),
we represent $\KK_0(\calg,\calg)$ as ${\rm End}(\K^\bullet(Z))$ and
take $\xi$ to be tensor product with the vector bundle $E$.
\begin{proposition}
\label{prop:KKfromcorr}
Given separable {\Ca}s $\alg$ and $\balg$, any class in
$\KK_d(\alg,\balg)$ comes from a noncommutative correspondence as
above, in fact from one with trivial $\xi=1_\calg$.
\end{proposition}
\begin{proof}
This is basically a restatement of a theorem of Cuntz, found
in~\cite[Corollary~17.8.4]{Black}. It suffices to take $d=0$,
since one can compose with the $\K$-oriented map $C_0(\bbr)\to\bbc$
which switches parity of degree. 
The theorem just cited asserts that given $x\in \KK(\alg,\balg)$, we
can write it in the form $x = f^*(j!)$ where $f\colon \alg
\to \calg$, $j$ is the inclusion of
the ideal in a split short exact sequence
\[
\xymatrix{
0~\ar[r] & ~\balg\otimes  \mathcal{K}~
\ar[r]^(.6)j & ~\calg~ \ar[r]_q & ~\dalg~ \ar@/_/[l]_s \ar[r] & ~0 \ ,
}
\]
and $j!$ is defined as in Example \ref{ex:quasihomo}. Then $x$ is
defined by the correspondence
\[
\xymatrix{& \calg & \\
\alg\ar[ur]^f  &   &    \balg\ar[ul]_{j\circ\iota} \ , }
\]
with $\iota\colon \balg \to \balg\otimes \mathcal K$ the usual 
stabilisation map.
\end{proof}
\begin{remark}
Proposition~\ref{prop:KKfromcorr} cannot be used as a
way of giving a new definition of $\KK$-theory groups, as we assumed
we already knew what a $\K$-oriented map is in the course of proving it.
We could get around this by starting only with $\KK$ classes associated
to $*$-homomorphisms and to split exact sequences. The effect would
then be to define $\KK$ as in~\S\ref{sec:KKaxioms} as the universal
quotient of the additive stable homotopy category of separable {\Ca}s
satisfying split exactness~\cite[Theorem~22.2.1]{Black}. (See
also~\cite{Nest} for another variant of this using triangulated
categories.)
\end{remark}

In this setting, we would like to describe the composition
product in $\KK$-theory. It is a morphism
\[
\KK(\alg, \dalg) \times \KK(\dalg, \balg) ~\longrightarrow~
\KK(\alg,\balg)
\]
which should be described as the composition of correspondences,
subject to some conditions as in the case of manifolds. Given the
correspondences
\[
\xymatrix{& \calg_1 &  & \calg_2 & \\
\alg\ar[ur]^{f_1}  &   &    \dalg\ar[ul]_{g_1} \ar[ur]^{f_2} & &
\balg \ar[ul]_{g_2} \ , }
\]
we would like to define their \emph{composition} as the correspondence
\[
\xymatrix{ 
&& \calg  && \\
& \calg_1 \ar@{.>}[ur]&& \calg_2 \ar@{.>}[ul]&\\
\alg\ar[ur]_{f_1}  \ar@{.>}@/^2pc/[uurr]^{f_\alg}&   & \dalg \ar[ul]^{g_1}
\ar[ur]_{f_2}  & & \balg\ar[ul]^{g_2} \ar@{.>}@/_2pc/[uull]_{g_\balg}}
\]
where $\calg$ is something like a suitable pushout or colimit
of the diagram
\beq
\xymatrix{\calg_1 & & \calg_2 \ . \\ & \dalg \ar[ul]^{g_1}
\ar[ur]_{f_2} & }
\label{pushoutdiag}\eeq
  
Part of the problem is that the notion of pushout for arbitrary
diagrams (\ref{pushoutdiag}) in the category of {\Ca}s is the
amalgamated {\Ca}ic free product $\calg_1 *_\dalg \calg_2$, the
universal {\Ca} generated by $*$-representations of $\calg_1$ and
$\calg_2$ on the same Hilbert space such that representations agree on
$g_1(d)$ and $f_2(d)$ for $d\in\dalg$. It is not obvious that the
natural map $\calg_2 \to \calg_1 *_\dalg \calg_2$ should be
$\K$-oriented just because $g_1\colon \dalg \to \calg_1$ is
$\K$-oriented, nor is it obvious that the composition $g_\balg$ of the
maps $g_2$ and $\calg_2\to \calg$ should be $\K$-oriented. Thus the
commutative case of~\S\ref{GeomCorr} above provides a guide. In that
case, we needed to use transversality to define the composition of
geometric correspondences, and the correct algebra
to take for $\calg$ was an amalgamated \emph{tensor} product
$\calg_1 \otimes_\dalg \calg_2$ which corresponds to the fibred
product of spaces.

A partial answer to this problem is discussed in~\cite[\S8.5]{CMR},
which addresses the problem of composing quasihomomorphisms.
This corresponds to the case where $f_2$ is the identity above, and
$g_1$ and $g_2$ are split injections of ideals. Some other cases of
interest are as follows. If $g_1$ is the identity, \emph{i.e.},
the first correspondence is simply a $*$-homomorphism ${f_1} :\alg
\rightarrow\dalg$, then we can compose $f_1$ and $f_2$ to get a
composition correspondence
\[
\xymatrix{ & \dalg & \\
\alg\ar[ur]^{f_2\circ f_1}  &   & \balg\ar[ul]_{g_2} \ . }
\]
 
\begin{example} 
Let $\alpha$ be an action of $\bbR$ on a {\Ca} $\alg$.
Connes' ``Thom isomorphism''~\cite{Connes-Thom,FS} (see
also~\cite[Theorem~19.3.6]{Black} and~\cite[Chapter~10]{CMR}) defines
a $\KK$-equivalence between the suspension of $\alg$, $\alg\otimes
C_0(\bbR)$, and the crossed product $\alg\rtimes_\alpha \bbR$. This
may be realised by noncommutative correspondences associated to two
Kasparov elements. One of them, in
$\KK_1(\alg,\alg\rtimes_\alpha\bbR)$, is described
in~\cite[\S19.3]{Black}. The inverse element, basically described
in~\cite[Chapter~10]{CMR}, is the Kasparov element in
$\KK_1(\alg\rtimes_\alpha\bbR,\alg)$ associated to the extension
\[
0~\longrightarrow~ \alg\otimes\mathcal{K} ~\longrightarrow~
\mathcal{W} ~\longrightarrow~ \alg\rtimes_\alpha \bbR
~\longrightarrow~ 0 \ ,
\]
where $\mathcal W$ is Rieffel's ``Weiner-Hopf extension
algebra''~\cite{Rief}. Combined with Takai duality, this provides a
primitive example of T-duality in the context of a noncommutative
orbifold spacetime.
\end{example} 

\begin{example}
Other noncommutative correspondences come from continuous trace  
$C^*$-algebras, discussed in~\S\ref{TwistedD} in the context of
D-branes in a background $H$-flux. They have been studied
in~\cite{RR}, and in the context of T-duality in~\cite{BEM,MR} and
other papers in the series by the same authors. We recall here the
simplest of these examples. Let
$$
\begin{CD}
\bbT @>>> \,  E \\
&& @VV \pi V \\
&& M \end{CD}
$$
be a principal circle bundle and $H$ a closed, integral three-form on
$E$. Then there is a continuous trace $C^*$-algebra $\CT(E, H)$ with
spectrum  equal to $E$ and Dixmier-Douady invariant equal to $[H] \in
\H^3(E, \bbZ)$.

Define $\widehat{E}$ to be the {T-dual} of $E$, which is another
oriented $\bbT$-bundle over $M$
\begin{equation}
\begin{CD}
\widehat{\bbT} @>>> \widehat{E} \\
&& @VV\widehat{\pi} V     \\
&& M \end{CD}
\end{equation}
with first Chern class $c_1(\,\widehat{E}\,) = \pi_* [H] $, where
$\pi_*  : \H^k(E,\bbZ) \to \H^{k-1}(M,\bbZ)$ denotes the pushforward
maps on cohomology.  The Gysin sequence for $E$ enables one to define
a T-dual $H$-flux $[\,\widehat{H}\,]\in H^3(\,\widehat{E},\bbZ)$
satisfying
$
c_1(E) = \widehat{\pi}_*[\, \widehat{H}\,] \,,
$
and such that $[H] = [\,\widehat{H}\,]$ in $\H^3(E\times_M \widehat{E},
\bbZ)$ where $E\times_M \widehat{E}$ is the fibred product.
Then there is a noncommutative correspondence
    $$
\xymatrix @=4pc { & \big(\CT(E\times_M \widehat{E}\,) \,,\, \xi\big) 
& \\
\CT(E, H)\ar[ur]^{f}  &   &    \CT\big(\widehat{E}\,,\, 
\widehat{H}\,\big)\ar[ul]_{g} }
$$
where $\xi$ is an analogue of the Poincar\'e line bundle, which  
determines an invertible element $$\alpha ~\in~ \KK_1\big(\CT(E, H)\,,\,
\CT(\,\widehat{E}, \widehat{H}\,)\big) \ , $$ and hence a
$\KK$-equivalence~\cite[\S6.2]{BMRS}.
\label{CTEx}\end{example}

Thus we see that bivariant K-theory is a convenient setting for
T-duality. But it also yields a potentially more refined version of
T-duality, which was originally presented as an isomorphism between
the K-theory groups of the spacetime $\alg$ and its T-dual $\balg$. In
Example~\ref{CTEx} above, $\alg= \CT(E, H)$ and $\balg =
\CT(\,\widehat{E}, \widehat{H}\,)$. On the other hand, when viewed as
a noncommutative correspondence, we have observed above and
in~\cite[\S6.2]{BMRS} that T-duality gives a $\KK$-equivalence $\alpha
\in \KK_1(\alg, \balg)$. We have seen that such a $\KK$-equivalence
determines  an isomorphism of K-theory groups in
${\sf{Isom}}(\K_\bullet(\alg), \K_{\bullet+1}(\balg))$. However, $
\alpha$ may contain more refined information, as the universal
coefficient theorem of Rosenberg and Schochet~\cite{RS87} states
that there is a split short exact sequence of abelian groups given by
\beq
\qquad 0 ~\to~{\Ext_\bbZ}\bigl(\K_{\bullet+1}(\alg)\,,\,
\K_\bullet(\balg)\bigr) ~\to~ \KK_\bullet\bigl(\alg\,,\,
\balg\bigr)~\to ~ {\Hom_\bbZ}\bigl(\K_\bullet(\alg)\,,\,
\K_\bullet(\balg) \bigr)~\to~ 0 \ .
\label{UCTKK}\eeq
This holds for a large class of $C^*$-algebras $\alg, \balg$ that
includes $\CT(E, H)$ and its T-duals. In the exact sequence
(\ref{UCTKK}), $\KK_\bullet(\alg, \balg)$ denotes the $\bbZ_2$-graded
group defined as the direct sum $\KK_\bullet(\alg, \balg)=\KK_0(\alg,
\balg) \oplus \KK_1(\alg, \balg)$, and $\K_\bullet(\alg)$ similarly
denotes $\bbZ_2$-graded groups. The first map (on the left) in the
exact sequence has degree 1 and the second map has degree 0. The
refined information when T-duality is viewed as a $\KK$-equivalence
$\alpha$ is contained in the group ${\Ext_\bbZ}(\K_{\bullet+1}(\alg),
\K_\bullet(\balg))$.

\begin{example}\label{BCC}
Let $X$ be a compact space, let $\alg$ be a unital {\Ca},
and let $\cV$ be a finitely generated projective module over
$C(X,\alg)\cong C(X)\otimes \alg$. By a slight modification of
the Swan-Serre theorem (noticed by Mishchenko and Fomenko), $\cV$
can be identified with the module of sections of an $\alg$-vector
bundle over $X$. (This is defined like an ordinary vector bundle,
except that $\alg$ plays the role of the scalars and the transition
functions of the bundle are required to be $\alg$-linear.) The
projective module $\cV$ gives a class
$[\cV]\in \K_0(C(X)\otimes \alg)=\KK(\bC, C(X)\otimes \alg)$.
If in addition $X$ is a closed spin$^c$ manifold, then $C(X)$
is a strong PD algebra, and so with a chosen fundamental class
$\Delta\in\KK(C(X)\otimes C(X),\complex)$ the composition product
$$[\cV]\otimes_{C(X)}\Delta~\in~ \KK\big(C(X)\,,\, \alg\big)$$ is
defined. This element can be viewed as coming from a very special kind
of noncommutative correspondence, essentially like that of
eq.~\eqref{commcorrdef} but with the ordinary bundle $E$ replaced by
the module $\cV$. We can also view this noncommutative correspondence
as
\beq
\label{eq:BCcorr}
\xymatrix @=4pc { & \big(C(X)\otimes \alg \,,\, \cV\big)
& \\
C(X) \ar[ur]^{\Id_{C(X)}\otimes 1}  &   &    \,\alg \ ,
\ar[ul]_{1\otimes\Id_\alg }}
\eeq
where the $\K$-orientation of the map on the right comes from the
spin$^c$ structure on $X$. Because the algebra $C(X)$ is commutative,
we can let $C(X)$ act on both sides on $\cV$ and thus view
$\cV$ not just as a class $[\cV]\in \KK(\bC,C(X)\otimes \alg)$
but also as a class $[[\cV]]\in \KK(C(X),C(X)\otimes \alg)$, as 
required for a noncommutative correspondence.

This situation arises frequently in connection with the Baum-Connes
conjecture. Let $\Gamma$ be a torsion-free discrete group such that
there is a model for its classifying space $B\Gamma$ which is a
compact spin$^c$ manifold. Then there is a canonical flat bundle
$\Upsilon :=  E\Gamma \times_\Gamma C^*_r (\Gamma)$ over $B\Gamma$
such that the space of all continuous sections $\cV:=
C_0(B\Gamma,\Upsilon)$ is a finitely generated projective
$C(B\Gamma)\otimes C^*_r(\Gamma)$-module. Then
\beq\label{BCcorr1}
\xymatrix @=4pc { & \big(C(B\Gamma) \otimes C^*_r(\Gamma) \,,\,
  [[\cV]]\big)
& \\
C(B\Gamma)\ar[ur]  &   &   C^*_r(\Gamma)  \ar[ul] }
\eeq
is a noncommutative correspondence of the sort (\ref{eq:BCcorr}). If
$\Gamma$ has the (strong) Dirac-dual Dirac property, then the
noncommutative correspondence (\ref{BCcorr1}) determines a $ \KK$-equivalence
between the algebras $C(B\Gamma)$ and $ C^*_r(\Gamma)$, \emph{i.e.},
the Baum-Connes conjecture is true. Groups $\Gamma$ that satisfy these
hypotheses include cocompact torsion-free discrete subgroups of
connected Lie groups whose noncompact semisimple part is $SO(n,1)$ or
$SU(n,1)$.
\end{example}

\begin{remark}
If algebras $\alg$ and $\balg$ are $\KK$-equivalent, then by the
diagram calculus for $\KK$-theory we can produce, by inflation, a
$\KK$-equivalence between $\calg\otimes\alg$ and $\calg\otimes\balg$
for any other algebra $\calg$. Starting from Examples~\ref{FMtransf}
and~\ref{BCC} above, one can thus deduce the parametrized
Fourier-Mukai transform and the parametrized Baum-Connes conjecture as
further examples of noncommutative correspondences.
\end{remark}

\end{document}